\documentclass[a4paper,UKenglish,cleveref, autoref, thm-restate]{lipics-v2021}

\hideLIPIcs  

\usepackage{mathtools}

\newtheorem*{problem*}{Problem}
\newtheorem*{openproblem*}{Open Problem}
\newtheorem*{conjecture*}{Conjecture}
\theoremstyle{claimstyle}

\newtheorem*{question*}{Question}

\usepackage{algorithm, algorithmic}

\newcommand{\LINECOMMENT}[1]{\item[] $\blacktriangleright$ \textit{#1}}

\title{RevCRN: Reversible Analog Computation using Chemical Reaction Networks} 

\author{Saptarshi Biswas}{Department of Computer Science, Iowa State University, Ames, IA, USA \and Ames National Laboratory, Ames, IA, USA}{sbiswas@iastate.edu}{https://orcid.org/0000-0001-7082-023X}{}

\author{James I. Lathrop}{Department of Computer Science, Iowa State University, Ames, IA, USA}{jil@iastate.edu}{}{}

\author{Rana D. Parshad}{Department of Mathematics, Iowa State University, Ames, IA, USA}{rparshad@iastate.edu}{}{}

\authorrunning{S. Biswas, J. I. Lathrop, and R. D. Parshad} 

\Copyright{Saptarshi Biswas, James I. Lathrop, and Rana D. Parshad} 



\acknowledgements{We would like to thank Dr. Jack Lutz and Dr. Titus Klinge for sharing their valuable feedback and insights on this work.}

\nolinenumbers 

\begin{document}

\maketitle

\begin{abstract}

    The computability of real numbers and functions using Turing Machines has been a central area of theoretical computer science since the mid-20th century. In the late 20th century, it was shown that chemical reactions can serve as a basis for computation using the Chemical Reaction Network (CRN) model. Recent advances in computing real numbers using Deterministic Chemical Reaction Networks (DCRNs) have identified numerous classes of DCRN-computable real numbers. In parallel, the works of R. Landauer and C. H. Bennett, spanning the 1960s to the early 2000s, showed that reversible computing offers significant advantages over irreversible methods, particularly in energy efficiency, motivating extensive research on reversible computation.

    In this work, we investigate the computability of real numbers using Reversible Chemical Reaction Networks (RevCRNs). The paper has two primary contributions: (1) establishing relationships among CRN-computable real number classes including Lyapunov CRN ($\mathbb{R}_{LCRN}$), Real-Time CRN ($\mathbb{R}_{RTCRN}$), rational numbers ($\mathbb{Q}$), and RevCRNs ($\mathbb{R}_{RevCRN}$), with key results: (i) $\mathbb{Q}$ is a strict subset of $\mathbb{R}_{RevCRN}$; (ii) the set of positive algebraic numbers ($ALG$), $\mathbb{R}_{LCRN}$, and real numbers computable by 1-species RevCRN ($\mathbb{R}_{RevCRN}^{1s}$) are equal; (iii) $\mathbb{R}_{RTCRN}$ and $\mathbb{R}_{RevCRN}$ exhibit non-empty overlap; and (iv) the set of real numbers computable by detailed-balanced RevCRNs ($\mathbb{R}^{DetBal}_{RevCRN}$) is a subset of $ALG$; and (2) exploring the existence of a hierarchy within $\mathbb{R}_{RevCRN}$. Finally, we leave open the exact relationship between $\mathbb{R}_{RevCRN}$ and $\mathbb{R}_{RTCRN}$ while conjecturing a general hierarchy of RevCRN-computable reals.
\end{abstract}

\section{Introduction}

Scientific research witnessed one of the historical shifts in the early 20th century, when Turing independently introduced the fundamental computational model, the Turing Machine, alongside Church's $\lambda$-calculus work \cite{Church_1936}, in his famous 1936 paper \cite{Turing1936}. This sparked a new branch of study encompassing the theory of computing real numbers and functions \cite{kawamura2009, KO1982, Pour-El_Richards_2017}.

On the other hand, the work of Shapiro and Shapley \cite{Shapiro_MassActionGibbsFunction_DCRNT_EquilibriumComputation}, in 1965, sheds light on the profound connection between the mass-action kinetics and the computability of the equilibrium concentrations of complex chemical reactions. In the same year, work by Aris introduced a field of study that we now call Chemical Reaction Network Theory \cite{Aris1965_DCRNT}. Since then, subsequent important works have established the field \cite{Feinberg1972_complex_balalnce_DCRNT, HornJackson1972_complex_balancing, Krambeck1970_DCRNT}. This has led to Chemical Reaction Networks (CRNs) emerging as a significant, unconventional computational paradigm and a key focus in molecular programming.

CRNs are based on designing chemical reactions to perform various tasks, including logical and arithmetic operations, as well as computing real numbers and functions \cite{Angluin_FastComputation_Probabilistic_PP, function_computability, Chen2014_CRN_function_computation, computability_dcrn, Almeida_SCRN_Turing_completeness}. Two of the most fundamental and widely used variant models of CRNs are Stochastic Chemical Reaction Networks (SCRNs) \cite{Cook2009_SCRN, Almeida_SCRN_Turing_completeness, SoloveichikCookWinfreeBruck_SCRN}, which are modeled using continuous-time Markov Chains, and Deterministic Chemical Reaction Networks (DCRNs), whose dynamics are governed by Ordinary Differential Equations (ODEs) \cite{Bournez2021_survey_analog, EpsteinPojman_CRN_book, computability_dcrn}. In this work, we will focus specifically on DCRNs.

CRNs serve as a special case \cite{Bournez2021_survey_analog, GPAC_RTCRN_equivalence, Huang2019} of Shannon's General Purpose Analog Computer \cite{Shannon_GPAC}. Fages, Guludec, Bournez, and Pouly, in their paper \cite{FagesGuludecBournezPouly_Turing_completeness}, have proven that DCRNs are Turing complete. In subsequent years, it has been shown that the set of real numbers computable by DCRNs includes both algebraic and transcendental numbers \cite{Fletcher2025_ALG_LCRN, huang_phd_thesis, GPAC_RTCRN_equivalence, Huang2019}. Huang et. al. in \cite{Huang2019} also introduce constrained models of DCRNs: the Real-Time Chemical Reaction Network (RTCRN) and the Lyapunov Chemical Reaction Network (LCRN), which satisfy conditions such as boundedness, exponential stability, and real-time computability. RTCRN and LCRN are formally defined in definitions \ref{def:rtcrn} and \ref{def:lcrn} in section \ref{sec:dcrn} of this paper. Further studies \cite{Fletcher2025_ALG_LCRN, Huang2019} highlight the relation between the sets of real numbers computed by different variants of DCRN. In this paper, we will consider $\mathbb{Q}$ as the set of all rational numbers, $ALG$ the set of all algebraic numbers, and $\mathbb{R}_{CRN}, \mathbb{R}_{LCRN}, \mathbb{R}_{RTCRN}$ the set of real numbers computable by deterministic CRN, LCRN, and RTCRN, respectively. Some key results include:
\begin{enumerate}
    \item $\mathbb{Q}\subsetneqq ALG\subsetneqq\mathbb{R}_{CRN}$ \cite{Huang2019}. 
    \item $\mathbb{R}_{LCRN}\subseteq\mathbb{R}_{RTCRN}$ \cite{Huang2019}. 
    \item $ALG=\mathbb{R}_{LCRN}$ \cite{Fletcher2025_ALG_LCRN}.
    \item Bournez, Fraigniaud, and Koegler showed that Large Population Protocols (LPP) \cite{bournez:hal-00760928, BOURNEZ20091340} can compute any algebraic number \cite{lpp_alg}. Later, Huang and Huls showed that any numbers computable by GPAC and CRN are computable by LPP \cite{huang_et_al:LIPIcs.DNA.28.7}.
\end{enumerate}

There have been studies on the stability of chemical reaction networks \cite{feinberg2019foundations}. Notably, there are significant theorems, such as the Deficiency Theorems \cite{feinberg2019foundations, Feinberg1972_complex_balalnce_DCRNT, FEINBERG198059_deficiency_one, FEINBERG19872229_deficiency_one, Feinberg1995_deficiency_one, Horn1972_deficiency_zero, HornJackson1972_complex_balancing}, which provide strong insights into the stability and convergence properties of reaction networks under specific conditions. Additionally, significant works introduce the concepts of detailed and complex balancing for analyzing mass flow equilibrium. \cite{Bridgman_detailed_balancing, Dirac_detailed_balancing, feinberg2019foundations, Fowler1925-pd_detailed_balancing, HornJackson1972_complex_balancing, Wegscheider1901_detailed_balancing}. 

However, in general, it is not guaranteed whether any arbitrary fully reversible or irreversible chemical reaction network converges to a stable equilibrium or if stable limit cycles exist to which the system may converge \cite{BOROS2023103839_limit_cycle, Erban2023_limit_cycle, SCHNAKENBERG1979389_limit_cycle}. As a consequence, since this work focuses on computing real numbers, we will concentrate exclusively on the family of reversible chemical reaction networks that have at least one stable fixed point on the positive real line.

C. H. Bennett's earlier works, spanning from 1973 to the early 2000s, explore the significance of reversible systems in the physical world and their implications for computing \cite{bennett_rev1,bennett_rev2,bennett_rev3}. From these works, as well as a standard textbook on chemical thermodynamics \cite{thermodynamics_book}, some key findings about reversible systems include:
\begin{enumerate}
    \item The net change in the system's entropy is low.
    \item The energy efficiency of such systems is high.
    \item The initial state of the system can be derived from any arbitrary state of the system by going backward through the process flow of the system.
\end{enumerate}

Recent works have also highlighted the significance of reversible systems across various known models of computation \cite{reversibility_reaction_system, BARYLSKA201848_reversible_petrinet, reversible_fsa, reversible_computing_models, reversibility_report}. Recent work by Kini and Doty \cite{kini2026reverserobustcomputationchemicalreaction} establishes that the reverse-robust CRN computation model can successfully compute any semilinear predicate or function. These works motivate us to extend reversible computing to the framework of CRNs computing real numbers. Besides energy efficiency, reversibility further enhances the cost-effectiveness of CRN systems by optimizing resource use, enabling species reuse, and facilitating economical extraction of species \cite{Hannah_reversibility_optimization}. In contrast, irreversible systems have difficulties recovering and extracting species after use. They also require reinitialization, a continuous supply of species for repeated computations, and exhibit greater heat generation and loss \cite{Landauer_irreversible_heat_generation}, which reversible systems can manage more efficiently. However, in this work we are not restricted to closed, fully reversible reaction networks that admit a thermodynamic equilibrium and thus exhibit energy efficiency. Instead, we also investigate open systems, where our interest lies not necessarily in energy efficiency but in the reusability of CRNs.

In this regard, this work provides a formal definition of the Reversible Chemical Reaction Network (RevCRN) and examines its advantages and computational limitations in the context of computing real numbers. The RevCRN model discussed in this study is based on fully reversible reactions considered within open, semi-open, and closed systems. The findings highlight the computational potential of the RevCRN framework, which could be used to optimize various computational processes, as previously mentioned. Precisely, the findings of this work suggest a relationship among the sets $\mathbb{Q}$, $ALG$, $\mathbb{R}_{LCRN}$, $\mathbb{R}_{RTCRN}$, and $\mathbb{R}_{RevCRN}$, which is illustrated in Figure \ref{fig:inclusion}, highlighting the overall scope of the class of RevCRN-computable reals $\mathbb{R}_{RevCRN}$.
\begin{figure}[htbp]
    \centering
    \includegraphics[width=0.55\linewidth]{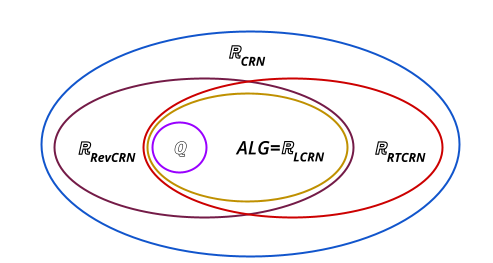}
    \caption{Inclusions of $\mathbb{R}_{RevCRN}$}
    \label{fig:inclusion}
\end{figure}

The key contributions of this research are summarized as follows:
\begin{enumerate}
    \item Establishing a broader idea about the set of all RevCRN-computable reals, denoted as $\mathbb{R}_{RevCRN}$, and its relation to other classes of CRN-computable reals, such as $\mathbb{Q}, ALG,$ $\mathbb{R}_{LCRN}$, and $\mathbb{R}_{RTCRN}$. 
    \item Investigate the existence of a hierarchy within $\mathbb{R}_{RevCRN}$.
\end{enumerate}

The remainder of this paper is organized as follows: Section 2 discusses the background theories and necessary definitions relevant to this research. Section 3 presents key results that contribute to the understanding of RevCRN computable reals and their overlap with the other classes of CRN-computable reals, namely, $\mathbb{Q}$, $ALG$, $\mathbb{R}_{LCRN}$, and $\mathbb{R}_{RTCRN}$. Section 4 focuses on the second main contribution of this work: the hierarchy of RevCRN-computable reals. Finally, Section 5 offers a brief discussion of the results and highlights the future research directions.

\section{Preliminaries}

Before we continue our discussion of the background theories, we will clarify some essential notations for better understanding. We define the sets $\mathbb{S} \subseteq \mathbb{N}$ for stoichiometric coefficients and $\mathbb{K} \subseteq \mathbb{Q}_{\geq 0}$ for rate constants. The notation $\mathbb{R}_{>0}$ represents the set of real numbers strictly greater than 0. The same interpretation holds for $\mathbb{Q}$ or $\mathbb{N}$ when substituting $\mathbb{R}$. A subscript indicating $\geq 0$ signifies that the set includes numbers greater than or equal to 0. Furthermore, for a polynomial $A$, $supp(A)$ denotes its support.

\subsection{Deterministic Chemical Reaction Networks (DCRN)}
\label{sec:dcrn}
A Deterministic Chemical Reaction Network $\mathcal{N}$ is generally defined using a tuple of 2 entities $\mathcal{N}=(S, R)$ where $S$ is the finite set of species taking part in the reaction system and $R$ is the finite set of reactions in the network. An example of such a DCRN is
\[A+X\xrightarrow{k_1} Y,\quad A\xrightarrow{k_2} Y\]
Here, $X$ and $A$ are the reactant species while $Y$ is the product species. $k_1$ and $k_2$ are the rate constants for the respective reactions. The product side and the reactant side are called \textit{complexes} individually. For e.g., the above reaction network has the following set of \textit{complexes} $\mathcal{C}=\{A+X, A, Y\}$. There can also be complexes that do not involve any species. Such a complex, called a $0-$complex, represents interaction with the external source or sink and is represented by $\phi$. The governing rate equations of the mass-action system are represented using the following ODEs
\[\frac{dx}{dt}=-k_1xa\qquad \frac{da}{dt}=-k_2a-k_1xa\qquad \frac{dy}{dt}=k_1xa+k_2a\]
Note that each of the $x,y,$ and $a$ appearing in the equation is actually a concentration trajectory function of time. 

The DCRN model represents a multivariate Initial Value Problem (IVP) in which each species is initialized with an initial concentration, and the system evolves, performing the desired computation. In general, if we have an $n-$species system $S={X_1, \dots, X_n}$, then the system can be represented as a Polynomial Initial Value Problem (PIVP) \cite{GRACA_PIVP}. In such a system, we represent the species concentration by a vector $x(t)=[x_1(t),\dots, x_n(t)]$, the initial concentrations by the vector $x(0)=[x_1(0),\dots, x_n(0)]$ and the rate equation dynamics using the polynomial vector $\frac{dx}{dt}=[p_1(x)\dots, p_n(x)]$ where each of the $p_i(x)$ represents a polynomial of the $n-$species trajectory functions $x_1(t),\dots,x_n(t)$. In this paper, we will denote the concentration vector as $x = [x_1, \dots, x_n]$ and the equilibrium concentration as $x^* = [x^*_1, \dots, x^*_n]$, considering concentrations at any arbitrary time in general.

Generally, for CRNs computing real numbers, it is considered that the stoichiometric coefficients are natural numbers, the initial values of each of the species and the rate constants of every reaction will be reals greater than or equals 0, i.e., $\forall x_i(0),k\in\mathbb{R}_{\geq0}$ \cite{Fletcher2025_ALG_LCRN, GPAC_RTCRN_equivalence, Huang2019, Chen2014_CRN_function_computation, feinberg2019foundations}. In this work, we aim to ensure fairness in the computation process by restricting initial concentrations and rate constants to positive rational numbers, i.e., $\forall x_i(0)\in\mathbb{Q}_{\geq 0}$ and $k\in\mathbb{Q}_{\geq 0}$. This approach is intended to prevent any unfair advantage that may arise from incorporating irrational numbers into rate constants or initial concentrations. Also note that all species concentrations at any given time are lower-bounded by 0. The paper \cite{Huang2019} introduces variants of DCRNs with additional constraints that we are interested in in this work. We restate the standard definitions without significant modification, following \cite{Fletcher2025_ALG_LCRN, GPAC_RTCRN_equivalence, Huang2019}. The norm appearing in the definitions is the standard $L_\infty$ norm \cite{Horn_Johnson_2012}.

\begin{definition}[Real-Time Chemical Reaction Network (RTCRN) computable reals \cite{Huang2019}] 
A real number $\alpha\in \mathbb{R}$ is said to be real-time computable by a CRN if there exists a CRN $\mathcal{N}=(S, R)$ and a designated output species $X_j\in S$ satisfying the following properties:
\begin{enumerate} 
    \item \textbf{Integrality:} The rate constant $k(\rho)$ of every reaction $\rho \in R$ is a positive integer ($k(\rho) \in \mathbb{Z}^+$). 
    \item \textbf{Boundedness:} There exists a constant $\beta\in \mathbb{R}_{>0}$ such that, when the system is initialized with the zero vector $x(0)=\mathbf{0}$, the concentration $x_i(t) \leq \beta$ for all species $X_i \in S$ and all $t\in[0,\infty)$. 
    \item \textbf{Real-Time Convergence:} If the system is initialized with the zero vector $x(0)=\mathbf{0}$, then for all $t\geq 1$, the trajectory function of the designated output species satisfies $|x_j(t)-|\alpha||\leq 2^{-t}$. 
\end{enumerate} 
\label{def:rtcrn} 
\end{definition}

To understand another variant of the CRN, the Lyapunov CRN, we first need to restate the definition of an exponentially stable point from \cite{Fletcher2025_ALG_LCRN, GPAC_RTCRN_equivalence, Huang2019}, as follows.

\begin{definition}[Exponentially Stable Point \cite{Huang2019}] 
    Let us consider a CRN with $n$ species. A fixed point $x^*\in\mathbb{R}^n_{\geq 0}$ is called an exponentially stable state if there exist constants $\delta, \alpha, C \in (0,\infty)$ such that, for any initial state $x(0)\in[0,\infty)^n$ satisfying $|x(0)-x^*| \leq \delta$, the system satisfies
    \begin{equation}
    |x(t)-x^*| \leq C e^{-\alpha t} |x(0)-x^*|
    \end{equation}
    for all $t\in[0,\infty)$.
    \label{def:exp_stable_pt} 
\end{definition}

We now restate the definition of Lyapunov CRN from \cite{Fletcher2025_ALG_LCRN, GPAC_RTCRN_equivalence, Huang2019}.

\begin{definition}[Lyapunov Chemical Reaction Network (LCRN) computable reals \cite{Huang2019}] 
    A real number $\alpha\in\mathbb{R}$ is said to be Lyapunov computable by a CRN if there exists a CRN $\mathcal{N}=(S, R)$, a designated output species $X_j\in S$, and a fixed point $x^*\in[0,\infty)^{|S|}$, where $|S|$ is the cardinality of set $S$, satisfying $x^*_j=|\alpha|$, such that the following properties hold:
    \begin{enumerate} 
        \item \textbf{Integrality:} The rate constants satisfy the integrality condition specified in Definition~\ref{def:rtcrn}.
        \item \textbf{Boundedness:} The species concentrations satisfy the boundedness condition specified in Definition~\ref{def:rtcrn}.
        \item \textbf{Exponential Stability:} The state $x^*$ is an exponentially stable state of the CRN $\mathcal{N}$ in the sense of Definition~\ref{def:exp_stable_pt}.
        \item \textbf{Convergence:} If the CRN is initialized with the zero vector $x(0)=\mathbf{0}$, then the trajectory of $\mathcal{N}$ converges to the stable state, meaning $\lim_{t\rightarrow\infty} x(t) = x^*$.
    \end{enumerate}
    \label{def:lcrn}
\end{definition}

\subsection{Open, semi-open and closed systems}

For simplicity, chemical reactions are commonly studied in systems where mass exchange with external sources or sinks is prohibited, thereby satisfying the law of conservation of mass within the reaction system. This kind of system is known as \textit{closed system} with respect to matter. However, by mathematically modeling chemical reaction systems, the mass flow can be generalized by accounting for mass exchange with external reservoirs. This is primarily done by considering pseudo-reactions of this form \cite{CONRADI2019279_CRN_book2,feinberg2019foundations}
\[\phi\rightleftharpoons A\]
Here, the $0-$complex $\phi$ is considered as the external reservoir and $A$ is a species of the reaction network. Note that even in these systems, the law of conservation of mass is conserved. This kind of network does not mean that mass is created or destroyed. Rather, it means that there is a mass flow between the system and the environment. The law of conservation of mass thus holds for both the system and the environment taken together. If the mass flow is considered in both directions, i.e., to and from the system, at the same time, then such a system is called \textit{open system} with respect to matter. Otherwise, it is called a \textit{semi-open system} if only one directional mass flow is considered \cite{CONRADI2019279_CRN_book2}.

\subsection{Irreversible, reversible, and weakly-reversible systems}

According to the standard definitions, the common types of reactions are defined as follows.

\begin{definition}[\cite{roberts2008chemical}]
    A reaction is considered reversible if a set of reactants can react to form a set of products, and those products can then react to regenerate the original reactants simultaneously. In contrast, a reaction is considered irreversible if there is no pathway to convert the products back into the original reactants.
\end{definition}

An example of such a reaction is
\begin{equation*}
    \begin{split}
        &Reversible:2H_2+O_2\rightleftharpoons 2H_2O\\
        &Irreversible:  HCl+NaOH\rightarrow NaCL+H_2O
    \end{split}
\end{equation*}

Feinberg, in the book \cite{feinberg2019foundations}, generalizes the standard definition of these systems in the following form.

\begin{definition}[\cite{feinberg2019foundations}]
    Let $y$ and $y'$ be any two complexes in a reaction network. Then the reaction between complexes $y$ and $y'$ is said to be reversible if both complexes react simultaneously to give each other. 
    
    Conversely, a reaction between the two complexes is considered weakly reversible if a reaction pathway exists between them, rather than a direct, simultaneous reaction.

    Any reaction with one-directional flow between the complexes is considered irreversible.
\end{definition}

It is important to note that all reversible reactions are inherently weakly reversible, but not all weakly reversible reactions are necessarily reversible. In this study, we will explore reversible systems defined within open, semi-open, and closed systems. In both open and closed systems, a backward reaction will always occur whenever a forward reaction takes place between complexes with positive rate constants. But in semi-open systems, we allow for one-directional exchange of matter at the boundary. This means that matter can enter or exit the system without requiring a strict reverse flow. However, any reaction occurring between two non-empty complexes must still have a corresponding reverse reaction with positive rate constants.

\subsection{Reversible CRN (RevCRN)}

In this work, we define the Reversible Chemical Reaction Network (RevCRN) as follows.

\begin{definition}[Reversible Chemical Reaction Network (RevCRN)]
\label{def:revcrn}
    A CRN $\mathcal{N}=(S, R)$ is considered to be a Reversible CRN if it obeys the following conditions.
    \begin{enumerate}
        \item \textbf{Rational Constraint:} All the rate constants and initial concentrations of species are positive rationals.
        \item \textbf{Natural Constraint:} All the stoichiometric coefficients are natural numbers greater than or equal to 0
        \item \textbf{Reversibility:}
        \begin{enumerate}
            \item \textbf{Open and Closed Systems:}\\
            Every forward reaction is accompanied by a reverse reaction (including the reaction with $0-$complex), each with a positive rate constant. In other words, if $k_f$ and $k_b$ represent the rate constants for the forward and backward reactions, then $k_f>0$ implies that $k_b>0$, and vice versa.
            \item \textbf{Semi-open Systems:}\\
            The system follows the same conditions as open and closed systems, except for unidirectional matter exchanges across its boundaries. When a nonempty complex in the reaction network interacts with a $0-$complex \(\phi\) that represents an external source or sink, the network can exhibit one-directional reaction flow without requiring a reverse flow. In other words, if \(k_f > 0\) in such cases, it is permissible for \(k_b = 0\), and vice versa.
        \end{enumerate}
    \end{enumerate}
\end{definition}

\begin{note}
    The rational constraint discussed here is considered without loss of generality, in contrast to definitions \ref{def:rtcrn} and \ref{def:lcrn}. Relaxing the integrality constraints to rational ones does not affect the computational power of the CRN. This approach ensures that we do not encode any irrational solutions into the rate constants or the initial conditions, preventing the CRN from having an unfair computational advantage.
\end{note}

In general, an $n-$species reversible system with set of species $S=\{X_1,\dots, X_n\}$ and set of reactions $R=\{R_1,\dots,R_k\}$ is given as
\begin{equation}
    i^{th}\ reaction\ R_i:\qquad a_{1i}X_1+\dots + a_{ni}X_n\xrightleftharpoons[k_{i2}]{k_{i1}} b_{1i}X_1+\dots + b_{ni}X_n
    \label{eq:revcrn_reaction_equation}
\end{equation}
Here, $a_{ji},b_{ji}$ for $j\in\mathbb{N}$ are stoichiometric coefficients of the reactant and product associated with the $j^{th}$ species in the $i^{th}$ reaction. Note that if a reactant or a product complex has all its stoichiometric coefficients equal to 0, then it corresponds to the $0-$complex $\phi$.

The general rate equation for any species $X_m\in S$ in such an $n-$species reversible mass-action system will be
\begin{equation}
    \begin{split}
        \frac{dx_m}{dt}=\sum_{i=1}^{|R|}(b_{mi}-a_{mi})\Bigg(k_{i1}\prod_{j=1}^{|S|}x_j^{a_{ji}}-k_{i2}\prod_{j=1}^{|S|}x_j^{b_{ji}}\Bigg),\ where\ b_{mi}>a_{mi}
    \end{split}
    \label{eq:rev_rate_general_nspecies}
\end{equation}

In the above equation, $|R|$ and $|S|$ are the cardinalities of sets $R$ and $S$, respectively. Further details about the general form of reversible reaction networks with one, two, and three species can be found in section \ref{sec:general_eq_123_sp} of the appendix.

Now we formally define the classes of real numbers computable by reversible chemical reaction networks (RevCRNs). Let $x_i(t)$ denote the concentration of species $X_i \in S$ at time $t$. We introduce the three structural variants of computing networks simultaneously below.



\begin{definition}[Open, Semi-Open, and Closed RevCRN Computable Reals] 
Let $\alpha \in \mathbb{R}_{\geq 0}$. We say that $\alpha$ is an \textbf{open} (resp., \textbf{semi-open}, \textbf{closed}) \textbf{RevCRN computable real} if there exists an open (resp., semi-open, closed) system RevCRN $\mathcal{N}=(S,R)$ operating under mass-action kinetics that has a designated output species $X_i \in S$ and a stable fixed point $x^* \in [0,\infty)^{|S|}$ satisfying $x^*_i = \alpha$, such that the following properties hold:
\begin{enumerate} 
    \item \textbf{Natural and Rational Constraint:} The structural parameters of $\mathcal{N}$ satisfy the conditions specified in Definition~\ref{def:revcrn}. 
    \item \textbf{Boundedness:} There exists a constant $\beta \in \mathbb{R}_{>0}$ such that all species concentration trajectories are bounded above by $\beta$, meaning $x_j(t) \leq \beta$ for all $X_j \in S$ and all $t \geq 0$. 
    \item \textbf{Stability:} The fixed point $x^*$ is a stable state of $\mathcal{N}$. 
    \item \textbf{Convergence:} When the network is initialized with an appropriate rational initial vector $x(0) \in \mathbb{Q}_{\geq 0}^{|S|}$, the species concentration trajectory of $\mathcal{N}$ satisfies $\lim_{t \rightarrow \infty} x(t) = x^*$. 
\end{enumerate} 
The sets of all real numbers computable by open, semi-open, and closed RevCRN systems are denoted by $\mathbb{R}^{\mathrm{O}}_{\mathrm{RevCRN}}$, $\mathbb{R}^{\mathrm{SO}}_{\mathrm{RevCRN}}$, and $\mathbb{R}^{\mathrm{C}}_{\mathrm{RevCRN}}$, respectively. 
\label{def:revcrn_variants} 
\end{definition}

\begin{definition}[RevCRN Computabile Reals]
A non-negative real number $\alpha \in \mathbb{R}_{\geq 0}$ is said to be a \textbf{RevCRN computable real} if:
\begin{equation}
\alpha \in \mathbb{R}_{\mathrm{RevCRN}}^{\mathrm{O}} \cup \mathbb{R}_{\mathrm{RevCRN}}^{\mathrm{SO}} \cup \mathbb{R}_{\mathrm{RevCRN}}^{\mathrm{C}}=\mathbb{R}_{\mathrm{RevCRN}}
\end{equation}
where $\mathbb{R}_{\mathrm{RevCRN}}$ denotes the complete set of all real numbers computable by RevCRNs.
\label{def:revcrn_general}
\end{definition}


\begin{definition}[$n$-species RevCRN Computable Reals] 
Let $\alpha \in \mathbb{R}_{\geq 0}$. The number $\alpha$ is said to be \textbf{$n$-species RevCRN computable} if there exists a RevCRN $\mathcal{N}=(S,R)$ such that $\mathcal{N}$ and $\alpha$ satisfy the conditions of Definitions~\ref{def:revcrn_variants} and \ref{def:revcrn_general}, with the total number of species exactly being $|S|=n$. The set of all real numbers computable by $n$-species RevCRNs is denoted by $\mathbb{R}_{\mathrm{RevCRN}}^{ns}$. 
\label{def:n_species_revcrn} 
\end{definition}

\section{Computing Real Numbers using RevCRNs}

This section presents the findings on the RevCRN-computable reals, which represent one of the two key contributions of this work. The results provide a comprehensive overview of the relationship between $\mathbb{R}_{RevCRN}$ and other established classes of computable real numbers, such as $\mathbb{Q}$, $ALG$, $\mathbb{R}_{LCRN}$, and $\mathbb{R}_{RTCRN}$. Additionally, this section emphasizes the inherent boundedness property of RevCRNs.

Considering Definition \ref{def:n_species_revcrn}, we will consider the classes of real numbers $\mathbb{R}_{RevCRN}^{ns}$ and $\mathbb{R}_{RevCRN}^{(n+1)s}$ to begin our discussion with the following proposition.

\begin{proposition}
$\mathbb{R}_{RevCRN}^{ns}\subseteq\mathbb{R}_{RevCRN}^{(n+1)s},\ \forall n\in \mathbb{N}_{>0}$
\label{prop:nested_inclusion_chain}
\end{proposition}
\begin{proof}
We establish the general nested inclusion chain by induction on the number of species $n$. Appendix~\ref{sec:general_eq_123_sp} provides an explicit analysis of the rate equations for one, two, and three species RevCRNs, demonstrating that $\mathbb{R}_{\mathrm{RevCRN}}^{1\mathrm{s}} \subseteq \mathbb{R}_{\mathrm{RevCRN}}^{2\mathrm{s}} \subseteq \mathbb{R}_{\mathrm{RevCRN}}^{3\mathrm{s}}$. 

To prove the general inductive step $\mathbb{R}_{\mathrm{RevCRN}}^{n\mathrm{s}} \subseteq \mathbb{R}_{\mathrm{RevCRN}}^{(n+1)\mathrm{s}}$ for any $n \in \mathbb{N}_{>0}$, let $\alpha \in \mathbb{R}_{\mathrm{RevCRN}}^{n\mathrm{s}}$. By Definition~\ref{def:n_species_revcrn}, there exists an $n$-species RevCRN $\mathcal{N} = (S, R)$ with $|S|=n$ that computes $\alpha$. We construct an equivalent $(n+1)$-species network $\mathcal{N}' = (S', R')$ by adding a catalytic species to the species set, $S' = S \cup \{X_{n+1}\}$. For every reaction $\rho \in R$, we define a corresponding reaction $\rho' \in R'$ where $X_{n+1}$ acts as a catalyst. 

Initializing the system with $x_{n+1}(0) = 1 \in \mathbb{Q}_{\geq 0}$, with the mass-action kinetics guaranteeing that $\frac{dx_{n+1}}{dt} = 0$, we get $x_{n+1}(t)=1$ for all $t \geq 0$. As a consequence, the addition of the new catalytic species to the system leaves the dynamics of the original $n$ species RevCRN unaltered. As shown in Appendix~\ref{sec:3s_revcrn} for the three-species case, $\mathcal{N}'$ inherits the exact stability and convergence properties of $\mathcal{N}$. Since $|S'| = n+1$, it follows that $\alpha \in \mathbb{R}_{\mathrm{RevCRN}}^{(n+1)\mathrm{s}}$, which completes the proof.
\end{proof}

\begin{proposition} 
\label{prop:union_equivalence}
$\bigcup_{k=1}^\infty\mathbb{R}_{\mathrm{RevCRN}}^{ks}=\mathbb{R}_{\mathrm{RevCRN}}$ 
\end{proposition} 

\begin{proof}
We establish the identity by demonstrating mutual set inclusion between the sets of RevCRN computable reals $\bigcup_{k=1}^\infty\mathbb{R}_{\mathrm{RevCRN}}^{ks}$, and $\mathbb{R}_{\mathrm{RevCRN}}$.

\smallskip\noindent
\textbf{($\subseteq$)} Let $\alpha \in \bigcup_{k=1}^\infty\mathbb{R}_{\mathrm{RevCRN}}^{ks}$. By the construction of $\bigcup_{k=1}^\infty\mathbb{R}_{\mathrm{RevCRN}}^{ks}$, we know that there must exist an integer index $n \in \mathbb{N}_{>0}$ such that $\alpha \in \mathbb{R}_{\mathrm{RevCRN}}^{ns}$. According to Definition~\ref{def:n_species_revcrn}, this implies the existence of a valid RevCRN $\mathcal{N}$ that satisfies the criteria of Definitions~\ref{def:revcrn_variants} and \ref{def:revcrn_general}. As a result, $\alpha$ is a RevCRN computable real, yielding $\alpha \in \mathbb{R}_{\mathrm{RevCRN}}$.

\smallskip\noindent
\textbf{($\supseteq$)} Conversely, let $\alpha \in \mathbb{R}_{\mathrm{RevCRN}}$. By Definition~\ref{def:revcrn_general}, $\alpha$ is computed by some underlying RevCRN $\mathcal{N}=(S,R)$. Because any physically or mathematically well-defined chemical reaction network is built upon a finite set of distinct species, the cardinality of its species set must be a finite integer, say $|S| = n \in \mathbb{N}_{>0}$. It follows immediately from Definition~\ref{def:n_species_revcrn} that $\alpha \in \mathbb{R}_{\mathrm{RevCRN}}^{ns}$. Since $\mathbb{R}_{\mathrm{RevCRN}}^{ns} \subseteq \bigcup_{k=1}^\infty\mathbb{R}_{\mathrm{RevCRN}}^{ks}$, we conclude that $\alpha \in \bigcup_{k=1}^\infty\mathbb{R}_{\mathrm{RevCRN}}^{ks}$.

Combining both inclusion directions completes the proof.
\end{proof}

\begin{lemma}[Boundedness of Species Concentration of any 2-species RevCRN]
    \label{lem:bound}
    Consider the general rate equation of the two-species CRN \eqref{eq:2sp_general}. Then the state variables $x$ and $y$ are bounded. That is, for any given positive initial condition $(x_{0},y_{0})$, there exists a time $t_{1}(x_{0},y_{0})$, and a finite constant $C$ (time-independent but depending on the problem parameters), such that for all time $t > t_{1}$,
    \[\sup X\leq C,\ \sup Y\leq C,\ where\ X=\{x(t)|\forall t\geq 0\}\ and\ Y=\{y(t)|\forall t\geq 0\}\]
\end{lemma}

\begin{proof}
     The main idea of this proof is based on the observation that, for a reference species $X_m$ in the general rate equation \ref{eq:rev_rate_general_nspecies}, the condition $b_{mi} > a_{mi}$ always holds. This condition essentially leads to logistic-type control for that species.

    To clarify, the general equation can be separated into two terms, after which a time-scaling can be applied via a change of time variable. By employing Young's inequality \cite{sell2002dynamics}, we can establish bounds on the growth and decay terms of the rate equation. Following this, proof by contradiction can be used, along with Grönwall's lemma \cite{perko2013differential}, to prove the statement of this lemma ultimately. The detailed proof is provided in Section \ref{proof:bound} of the appendix.
\end{proof}

\begin{corollary}[Boundedness of Species Concentration of any n-species RevCRN]
     \label{lem:bound11}
    Consider the general rate equation of the n-species CRN. Then the state variables $x_{i}$ are bounded. That is, for any given positive initial condition  there exists a time $t_{1}$, depending on the initial condition and a finite constant $C$ (time-independent but depending on the problem parameters), such that for all time $t > t_{1}$,
    $\sup X\leq C$, where \ $X=\ \lbrace x_{i}(t)|\forall t\geq 0\ , \forall i\rbrace\ $
\end{corollary}
\begin{proof}
The proof is a direct extension of Lemma \ref{lem:bound} to the n-species case, and follows the same structure.
\end{proof}

Specifically for the following Lemma \ref{lem:rational}, we will restrict the rational constraints of the RevCRN to natural constraints. This means that we will only consider natural numbers for the initial concentrations of species and for the rate constants. This approach is intended to prevent the situation where the desired solution could be encoded in the rate constants or the initial concentrations, giving the RevCRN an unfair advantage.

\begin{lemma}
    Fully open RevCRNs with only one species can compute all rational numbers. In particular, $\mathbb{Q}\subsetneqq\mathbb{R}^{1s}_{RevCRN}$.
    \label{lem:rational}
\end{lemma}

\begin{proof}

     Consider the following RevCRN
    \begin{equation}
        \begin{split}
            \phi&\xrightleftharpoons[p]{q} X
        \end{split}
        \label{revcrn:rationals}
    \end{equation}
    Here, $p,q\in\mathbb{N}$. This leads to the following ODE.
    \begin{equation*}
        \frac{dx}{dt}=p-qx
    \end{equation*}
    Considering $x(0)=0\in\mathbb{N}$, we get the following solution of the above ODE
    \begin{equation*}
        \begin{split}
            &x(t)=\frac{1}{q}(p-e^{-k_1(t+c)})\\
            \Rightarrow &\lim_{t\rightarrow\infty}x(t)=\frac{p}{q}
        \end{split}
    \end{equation*}

    Here, $c$ is the integral constant. This construction shows that we can create any rational number of the form $\frac{p}{q}$ where $p,q\in\mathbb{N}$.

    Let us now consider the following fully open reversible reaction.
    \begin{equation}
        \begin{split}
            \phi&\xrightleftharpoons[p]{q} 2X
        \end{split}
        \label{revcrn:irrationals}
    \end{equation}
    This changes the ODE to
    \begin{equation*}
        \frac{dx}{dt}=p-qx^2,\quad p,q\in\mathbb{N}
    \end{equation*}

    On analyzing this ODE, we can find that when steady state is achieved, i.e., $\frac{dx}{dt}=0$ as $t\rightarrow\infty$, we get $\lim_{t\rightarrow\infty}x(t)=\sqrt{\frac{p}{q}}$. This proves the lemma's statement.
\end{proof}

\begin{definition}[\cite{lang_algebra_2002}, \cite{marcus2018number}]
    \label{def:alg}
    A number $\alpha\in\mathbb{C}$ is said to be algebraic if it is a root of an irreducible polynomial of degree $n\in\mathbb{N}$ over $\mathbb{Q}$, i.e.,
    \[\mathcal{P}(\alpha)=a_0+a_1\alpha+\dots +a_n\alpha^n=0,\quad \forall i, a_i\in\mathbb{Q}\]
\end{definition}
In the context of CRNs, this study will only consider positive algebraic numbers in $\mathbb{R}$. We denote the set of all positive algebraic numbers in \(\mathbb{R}\) as \(ALG\) \cite{Fletcher2025_ALG_LCRN, Huang2019}.

For the following proposition, let \(\mathcal{P}(x)\) be an arbitrary single-variable polynomial of degree $n$ represented in the following form:
\[\mathcal{P}(x)=\sum_{i=0}^na_ix^i,\quad \forall a_i\in \mathbb{Q},a_n> 0\]

\begin{proposition}
   Consider the set \(root(\mathcal{P}) = \{\alpha \mid \alpha \in \mathbb{R}_{>0}, \mathcal{P}(\alpha) = 0\}\). If we construct another polynomial \(\mathcal{Q}(x) = \mathcal{P}(x)(k - x)\) for some \(k \in \mathbb{Q}\), then we have \(root(\mathcal{Q}) = root(\mathcal{P}) \cup \{k\}\). Furthermore, if \(k > \alpha\) for any \(\alpha \in root(\mathcal{P})\), then the stability remains invariant for all \(\alpha' < k\) where \(\alpha' \in root(\mathcal{P})\).
    \label{prop:stability_invariance}
\end{proposition}

\begin{proof}
    From the construction of $\mathcal{Q}(x)$ it is easy to observe that $\mathcal{Q}(x)=0$ whenever $x=k$ or $\mathcal{P}(x)=0$, i.e., $x\in root(\mathcal{P})$. This shows that $root(\mathcal{Q}) = root(\mathcal{P}) \cup \{k\}$.
    
    Now we proceed to do the stability analysis of roots of $\mathcal{Q}(x)$. We first compute the first-order derivative of $\mathcal{Q}(x)$ with respect to $x$.
    \[\mathcal{Q}'(x)=\frac{d\mathcal{Q}}{dx}=\frac{d}{dx}\bigg((k-x)\mathcal{P}(x)\bigg)=(k-x)\mathcal{P}'(x)-\mathcal{P}(x)\]
    It can be shown that for polynomial system $\mathcal{P}(x)$, the eigenvalues $\lambda_p=\mathcal{P}'(x^*)$ where $\mathcal{P}(x^*)=0$. Similarly, for the system represented by the polynomial $\mathcal{Q}(x)$, the eigenvalues at points $\forall x^*\in root(\mathcal{Q})$ are
    \[\lambda_q=\mathcal{Q}'(x^*)=(k-x^*)\mathcal{P}'(x^*)=(k-x^*)\lambda_p,\ \forall x^*\in root(\mathcal{P})\]
    Note that the positivity or negativity of eigenvalues is preserved when $x^*<k$ and altered when $x^*>k$. We know that $\forall x^*\in ALG$, $\exists k\in \mathbb{Q}: k>x^*$. This proves the proposition.
\end{proof}

\begin{lemma}
    $ALG\subseteq \mathbb{R}^{1s}_{RevCRN}$
    \label{lem:alg}
\end{lemma}

\begin{proof}
  Let \(\alpha \in \text{ALG}\) and \(\alpha > 0\). Based on the statement and proof of Lemma 5.1 in \cite{Huang2019}, we know that there exists a single-variable polynomial \(\mathcal{P}(x)\) with integral coefficients that describes the rate equation of a 1-species CRN, given by \(\frac{dx}{dt} = \mathcal{P}(x)\). In this context, \(\alpha\) is not only a root of \(\mathcal{P}(x)\), but it is also a stable root of the rate equation. Without loss of generality, we will consider rational coefficients in our case to align with the construction of RevCRN.

    The rate equation of a general 1-species CRN that computes algebraic numbers through a single-variable polynomial does not typically adhere to the form specified in equation~\ref{eq:rev_rate_general_nspecies}. However, it is possible to transform any such polynomial into one that obeys the form outlined in equation \ref{eq:rev_rate_general_nspecies}, while preserving the original polynomial's roots. To achieve this, we present Algorithm \ref{alg:poly_conversion} below, which is a crucial tool for this conversion.

    The algorithm transforms a single-variable polynomial of degree \( n \in \mathbb{N} \) with a negative leading coefficient and a positive constant term into a form suitable for the rate equation of a single species RevCRN, represented as \( \mathcal{P}_{RevCRN}(x) \) in equation \ref{eq:rev_rate_general_nspecies}.

Starting with an initial polynomial \( \mathcal{P}(x) = \sum_{i=0}^n a_i x^i \) (where \( a_0 \geq 0 \) and \( a_n < 0 \)), the algorithm organizes the terms in decreasing order of their exponents. It checks for alternating positive and negative coefficients, beginning with the negative leading term. For a 1-species RevCRN, polynomials with \( n \) terms (even \( n \)) can be expressed as \( \frac{n}{2} \) binomials of the form \( k_{i1}x^{a_i} - k_{i2}x^{b_i} \) under the condition that \( a_i < b_i \).

To convert \( \mathcal{P}(x) \) into \( \mathcal{P}_{RevCRN}(x) \), the algorithm performs two operations. First, it addresses any sign violations in intermediate terms by splitting them into positive and negative components, such that \( kx^a = (p - q)x^a \). This allows the terms to be rearranged for constructing the desired binomials until no violations remain.

An optional step involves closing exponent gaps that occur when terms like \( a_m x^m \) and \( a_n x^n \) have \( n > m + 1 \), indicating missing intermediate exponents. To resolve these gaps, dummy net-zero terms (e.g., \( x^k - x^k \) where \( m < k < n \)) can be introduced to create valid binomial terms.

    \begin{algorithm}[htb!]
        \caption{Polynomial Transformation}
        \label{alg:poly_conversion}
        \begin{algorithmic}[1]
            \REQUIRE $P(x) = \sum_{i=0}^{n} a_i x^{d_i}$ with $a_i \in \mathbb{Q}$
            \ENSURE Transformed polynomial with alternating signs (and optionally no exponent gaps)
            
            \STATE Arrange the terms of $P(x)$ in decreasing order of exponents, i.e., the new arrangement of the terms should follow $\forall 0< i< n, d_{i-1}>d_i>d_{i+1}$

            \LINECOMMENT{\textbf{Ensure alternating coefficient signs}}
            \FOR{each pair of consecutive non-constant terms $A = a x^{\alpha}$, $B = b x^{\beta}$ with $\alpha > \beta$}
            
                \IF[No Violation: Required alternating signs found]{$ab < 0$}
                    \STATE Continue to the next pair
                \ELSIF[Consecutive terms have negative coefficients]{$a < 0$ and $b < 0$}
                    \STATE Choose $p, q \in \mathbb{Q}$ such that $b = p - q$\COMMENT{Break term}
                    \STATE Replace $B$ by $p x^{\beta} - q x^{\beta}$\COMMENT{Reorganize as required}
                    \STATE Pair $p x^{\beta}$ with $A$, and continue from $-q x^{\beta}$
                \ELSIF[Consecutive terms have positive coefficients]{$a > 0$ and $b > 0$}
                    \STATE Choose $p, q \in \mathbb{Q}$ such that $b = p - q$\COMMENT{Break term}
                    \STATE Replace $B$ by $-q x^{\beta} + p x^{\beta}$\COMMENT{Reorganize as required}
                    \STATE Pair $-q x^{\beta}$ with $A$, and continue from $p x^{\beta}$
                \ENDIF
            \ENDFOR
            \STATE Repeat the above process until adjacent terms alternate in sign
            
            \LINECOMMENT{\textbf{Optional: Eliminate exponent gaps}}
            
            \WHILE{there exist adjacent terms $A = -p x^{\alpha}$ and $B = q x^{\beta}$ with $\alpha - \beta > 1$}
            
                \STATE Choose $c$ such that $\alpha > c > \beta$ (e.g., $c = \alpha - 1$)\COMMENT{Exponent term selection}
                \STATE Choose $k \in \mathbb{Q}$ \COMMENT{Valid coefficient selection}
                \STATE Insert the canceling pair $k x^c - k x^c$ between $A$ and $B$ \COMMENT{Net-zero term insertion}
                \STATE Pair $k x^c$ with $A$ and $-k x^c$ with $B$ \COMMENT{Rearrange terms as required}
            
            \ENDWHILE
            
            \RETURN resulting polynomial $\mathcal{Q}(x)$
        
        \end{algorithmic}
    \end{algorithm}
        
    Considering Algorithm \ref{alg:poly_conversion}, we distinguish two cases and handle them separately:
    \begin{enumerate}
        \item \textbf{Negative leading term: }  
        If the highest-degree term of the polynomial carries a negative coefficient, then one may directly apply Algorithm~\ref{alg:poly_conversion} to obtain a representation consistent with Equation~\ref{eq:rev_rate_general_nspecies}.
        
        \item \textbf{Positive leading term: }  
        If the highest-degree term carries a positive coefficient, consider the transformed polynomial
        \[\mathcal{Q}(x) = \mathcal{P}(x)(k - x),\]
        where $k \in \mathbb{Q}$ is chosen such that $k > \alpha$, and $\alpha \in \mathrm{root}(\mathcal{P})$ denotes the stable root of $\mathcal{P}$ under consideration.
        
        By Proposition~\ref{prop:stability_invariance}, this transformation preserves the stability properties of the root $\alpha$ without loss of generality. One can then apply Algorithm~\ref{alg:poly_conversion} to $\mathcal{Q}(x)$ to obtain a representation in the desired form of Equation~\ref{eq:rev_rate_general_nspecies}.
    \end{enumerate}
\end{proof}

\begin{note}
     The final polynomial generated by Algorithm \ref{alg:poly_conversion} represents the rate equation of a single species RevCRN. Note that the algorithm can generate isolated negative terms that are not paired with any positive terms. Recall that such systems are feasible under our definition of RevCRN, which represent semi-open systems.
\end{note}

\begin{example}
    Consider the following polynomial
    \[\mathcal{P}(x)=1-5x^3+x^4\]
    The polynomial represents the mass-action kinetics of the following CRN
    \[4X\xrightarrow{1}5X\qquad \phi\xrightleftharpoons[5/3]{1/3}3X\]
    The highest exponent term of $\mathcal{P}(x)$ has a positive coefficient. This represents the second case discussed above. As a result, we first transform $\mathcal{P}(x)$ into $\mathcal{Q}(x)$ considering $k=3$ for simplicity as follows:
    \[\mathcal{Q}(x)=\mathcal{P}(x)(3-x)=3-x-15x^3+8x^4-x^5\]
    We then apply Algorithm \ref{alg:poly_conversion} on $\mathcal{Q}(x)$ as follows:
    \begin{enumerate}
        \item \textbf{Sorting:}
        
        We first rearrange the terms of $\mathcal{Q}(x)$ in decreasing order of exponents as follows:
        \[\mathcal{Q}(x)=-x^5+8x^4-15x^3-x+3\]

        \item \textbf{Ensure alternating coefficient sign}

        Note that the alternating sign condition in the sequence is violated between the terms \(-15x^3\) and \(-x\). Since both terms carry negative coefficients satisfying the condition in line 5 of Algorithm \ref{alg:poly_conversion}, we express \(-x\) in the form of \(-x = (p - q)x\), where \(p - q = -1\). By setting \(p = 1\) and \(q = 2\), we obtain \(-x = x - 2x\). Next, we pair \(x\) with \(-15x^3\) and \(-2x\) with \(3\). Consequently, the polynomial can be represented as:
        \[\mathcal{Q}(x) = -x^5 + 8x^4 - 15x^3 + x - 2x + 3\]
        At this stage, \(\mathcal{Q}(x)\) is free from alternating sign violations and satisfies the desired form.

        \item \textbf{Eliminate exponent gaps}

        There exists a gap between the exponent terms $-15x^3+x$. So we insert the terms $kx^2-kx^2$ between them and choose $k=1$. Then the resulting polynomial generated is
        \[\mathcal{Q}(x)=-x^5+8x^4-15x^3+x^2-x^2+x-2x+3\]
    \end{enumerate}

    The polynomial $\mathcal{Q}(x)$ now represents the mass-action kinetics of the following RevCRN.
    \begin{equation*}
        \begin{split}
            4X\xrightleftharpoons[1]{8}5X \qquad 2X\xrightleftharpoons[15]{1}3X \qquad 2X\xrightleftharpoons[15]{1}3X \qquad X\xrightleftharpoons[1]{1}2X \qquad \phi \xrightleftharpoons[2]{3}X
        \end{split}
    \end{equation*}
\end{example}

\begin{corollary}
    $\mathbb{R}_{LCRN}=ALG=\mathbb{R}^{1s}_{RevCRN}$
    \label{cor:alg_equivalence}
\end{corollary}

\begin{proof}
    We know that \( ALG = \mathbb{R}_{LCRN} \) as established in \cite{Fletcher2025_ALG_LCRN}. From Lemma \ref{lem:alg}, we can conclude that \( ALG \subseteq \mathbb{R}^{1s}_{RevCRN} \). Additionally, it is straightforward to demonstrate that \( \mathbb{R}^{1s}_{RevCRN} \subseteq ALG \) because the rate equations of all 1-species RevCRNs are represented by single-variable polynomials with rational coefficients. Therefore, the corollary follows directly from these findings.
\end{proof}

\begin{proposition}
    If a RevCRN converges to a positive stable state, then its components are in $ALG$.
    \label{prop:pos_steady_state_ALG}
\end{proposition}

\begin{proof}
    From Theorem 3.3 of \cite{Fletcher2025_ALG_LCRN}, we know that any isolated fixed points of a CRN have all its components in $ALG$. It suffices to show that the positive stable states of RevCRN are isolated fixed points.
    
    While analyzing the positive steady states $x^*=[x_1^*,x_2^*,\dots, x_n^*]$ of a RevCRN, these equilibria are not always component-wise isolated in the full concentration space $\mathbb{R}_{>0}^n$. For instance, networks containing catalytic reactions in which the catalyst concentrations remain strictly invariant possess linear conservation laws, thereby inducing a continuum of equilibria. To analyze this structure, the full state space can be decomposed into two distinct orthogonal subspaces: the stoichiometric subspace, which represents the evolving species, and its orthogonal complement, which represents the invariant catalytic species. Consequently, while the steady states form a continuous manifold along the catalytic invariant dimensions, they typically reduce to isolated fixed points when restricted to any single stoichiometric compatibility class induced by the stoichiometric subspace of the evolving species.

    According to our construction, we only allow positive rational initial concentrations. As a result, the effect of a rational constant induced by a catalytic species over the reactions can be encoded within the rate equation. For e.g., consider the general $n-$species RevCRN following mass-action kinetics:
    \[a_1X_1+\dots+a_nX_n+Y\xrightleftharpoons[k_2]{k_1}b_1X_1+\dots+b_nX_n+Y,\quad a_i\neq b_i, 1\leq i\leq n\]
    The rate equation of the reaction is
    \[\frac{dy}{dt}=0,\qquad \frac{dx_i}{dt}=k_1y\prod_{i=1}^n x_i-k_2y\prod_{i=1}^nx_i^2,\quad 1\leq i\leq n\]
    As per our construction, we initialize $y(0)=q\in\mathbb{Q}$ and we know $y(t)=q,\ \forall t\geq 0$. This results in the rate equation of $x_i$ to be
    \[\frac{dx_i}{dt}=k_1q\prod_{i=1}^nx_i-k_2q\prod_{i=1}^nx_i^2=\overline{k_1}\prod_{i=1}^nx_i-\overline{k_2}\prod_{i=1}^nx_i^2\ where\ \overline{k_1}=k_1q,\overline{k_2}=k_2q\]
    The RevCRN can be thus converted to $a_1X_1+\dots+a_nX_n\xrightleftharpoons[\overline{k_2}]{\overline{k_1}}b_1X_1+\dots+b_nX_n$
    To prevent this reaction from occurring when \(y(t) = 0\) for all \(t \geq 0\), one can simply set \(q = 0\). This modification results in \(\overline{k_1} = \overline{k_2} = 0\), thereby facilitating the desired reaction behavior.

    As a result, we can reduce our RevCRN to a non-degenerate stoichiometric subspace of evolving species. Lastly, it remains to show that the fixed point in the subspace induced by the update is isolated. Recall that we are considering RevCRNs represented by a system of ODEs $\frac{dx}{dt}=\mathcal{F}(x)$ (of form equation \ref{eq:rev_rate_general_nspecies}), where $x=[x_1,x_2,\dots, x_n]$ that are bounded and converge to positive real numbers (including 0). In particular, for stable states in such reduced systems discussed above, we have $||\mathcal{F}'(x^*)||_\infty<0$. Furthermore, since $\mathcal{F}(x)$ is continuous, $\exists\delta>0$ such that whenever $||x^*-x||_\infty<\delta$, $\lim_{t\rightarrow\infty} x(t)= x^*$. This shows that the point $x^*$ of the reduced system is an isolated fixed point.
\end{proof}

    Let us consider any two complexes $ y$ and $ y'$ in a fully reversible reaction network $\mathcal{N}=(S, R)$ following mass-action kinetics.
    \[y\xrightleftharpoons[k_{y'\rightarrow y}]{k_{y\rightarrow y'}}y'\]
    
    In general, the complexes of an $n-$species reaction network are of the form $y=a_1X_1+\dots+a_nX_n$ and $y'=b_1X_1+\dots+b_nX_n$. The mass action kinetics of the reaction $y\xrightarrow{k_{y\rightarrow y'}}y'$ is defined as 
    \[\mathcal{K}_{y\rightarrow y'}(x)=k_{y\rightarrow y'}\prod_{i=1}^nx_i^{a_i},\ where\ x=[x_1,\dots,x_n]\]
    Similarly, for the reaction $y'\xrightarrow{k_{y'\rightarrow y}}y$, we have $\mathcal{K}_{y'\rightarrow y}(x)=k_{y'\rightarrow y}\prod_{i=1}^nx_i^{b_i}$. 
    
    Now we define detailed balancing and detailed balanced RevCRN computable reals as follows.

\begin{definition}[Detailed balancing \cite{Bridgman_detailed_balancing, Dirac_detailed_balancing, feinberg2019foundations, Fowler1925-pd_detailed_balancing, Wegscheider1901_detailed_balancing}]
\label{def:detailed_balancing}
    Consider a fully reversible reaction network $\mathcal{N}=(S,R)$ following mass action kinetics. \textbf{Detailed balancing} is said to be obtained at concentration $x^*\in\mathbb{R}_{>0}^n$, if the following condition is satisfied:
    \[\mathcal{K}_{y\rightarrow y'}(x^*)=\mathcal{K}_{y'\rightarrow y}(x^*),\qquad \forall y\rightarrow y'\in R\]
\end{definition}

\begin{definition}[Detailed Balanced RevCRN Computable Reals] 
    Let $\alpha \in \mathbb{R}$. The real number $\alpha$ is said to be a \textbf{detailed balanced RevCRN computable real} if there exists a network $\mathcal{N}=(S,R)$ satisfying the conditions of Definition~\ref{def:revcrn_variants} such that $\mathcal{N}$ admits a detailed balanced equilibrium and the designated output species concentration satisfies $\lim_{t \rightarrow \infty} x_i(t) = |\alpha|$. The set of all real numbers computable by detailed balanced RevCRNs is denoted by $\mathbb{R}_{\mathrm{RevCRN}}^{DetBal}$.
    \label{def:detailed_balanced_revcrn}
\end{definition}

Considering Definition \ref{def:detailed_balancing}, and \ref{def:detailed_balanced_revcrn}, we present the following corollary.
    
\begin{corollary}
    $\mathbb{R}_{RevCRN}^{DetBal} \subseteq ALG$.
\end{corollary}

\begin{proof}
   According to the Definition \ref{def:detailed_balancing}, a detailed balanced system is one in which positive steady states are established when the rates of the forward and reverse reactions are equal \cite{Bridgman_detailed_balancing, Dirac_detailed_balancing, feinberg2019foundations, Fowler1925-pd_detailed_balancing, Wegscheider1901_detailed_balancing}. As a result, the corollary follows directly from Proposition \ref{prop:pos_steady_state_ALG}.
\end{proof}

\section{$\mathbf{\mathbb{R}_{RevCRN}}$ Hierarchy}
\label{sec:hierarchy}

From Corollary \ref{cor:alg_equivalence}, we know that all the real numbers computed by the 1-species RevCRN are algebraic. However, we observe a critical result when studying the 2-species semi-open RevCRN. This leads to the following proposition, which serves as the motivation of this section.

\begin{proposition}
    $\mathbb{R}_{RevCRN}\neq ALG$
    \label{prop:ALG_RevCRN_nequiv}
\end{proposition}
\begin{proof}
From Corollary \ref{cor:alg_equivalence}, we know that $ALG=\mathbb{R}_{RevCRN}^{1s}\subseteq \mathbb{R}_{RevCRN}$. Now it suffices to show that there exists at least a 2-species CRN that computes a number not in $ALG$. Consider the following RevCRN.
\begin{equation}
    \begin{split}
        X\xrightleftharpoons[k_{2}]{k_{1}}\phi\qquad
        X+Y\xrightleftharpoons[k_{4}]{k_{3}}X+Y+Y
    \end{split}
\end{equation}

Its corresponding ODEs representing the mass-action kinetics are
\begin{equation}
    \begin{split}
        \frac{dx}{dt}=k_2-k_1x\qquad
        \frac{dy}{dt}=k_3xy-k_4xy^2
    \end{split}
\end{equation}

In compliance with the definition of a semi-open RevCRN system, we allow $k_2=0$. Then we get by solving the two ODEs
\begin{equation}
\begin{split}
    x(t)=e^{-k_1t}\qquad
    y(t)=\frac{k_3}{k_4\pm\exp{\Big(-\frac{k_3}{k_1}\big(1-\exp{(-k_1t)}\big)\Big)}}
\end{split}
\end{equation}

We can now notice that as $t\rightarrow\infty$ we get
\begin{equation}
    \begin{split}
        \lim_{t\rightarrow\infty}x(t)=0\qquad
        \lim_{t\rightarrow\infty}y(t)=\frac{k_3\exp{(\frac{k_3}{k_1})}}{k_3+k_4\big(\exp{(\frac{k_3}{k_1})}- 1\big)}=y_\infty
    \end{split}
\end{equation}

We can observe that $\exists k_1,k_3,k_4\in\mathbb{Q}_{>0}$ such that $y_\infty$ is a transcendental number.

It suffices to verify the claim with a single example. Choose $k_1=1, k_3=2, k_4=1$. In this case, we obtain
\[
y_\infty=\frac{2e^2}{1+e^2}.
\]
To prove that this value is transcendental, observe first that, by the Lindemann–Weierstrass theorem \cite{baker2022transcendental}, $e^2$ is transcendental. Rearranging the formula for $y_\infty$ yields
\[
e^2=\frac{y_\infty}{2-y_\infty}.
\]
If $y_\infty$ were algebraic, then applying rational operations with other algebraic numbers would again produce an algebraic number, which would force $e^2$ to be algebraic. This contradiction shows that $y_\infty$ must be transcendental.
\end{proof}

\begin{corollary}
    If a RevCRN with more than one species converges to a non-positive equilibrium, then its components may not belong to $ALG$.
\end{corollary}

\begin{proof}
    The example provided in Proposition \ref{prop:ALG_RevCRN_nequiv} shows that there exist RevCRNs that can compute transcendental numbers where one of the species converges to 0 at the limit.
\end{proof}

\begin{corollary}
    $\mathbb{R}^{1s}_{RevCRN}\subsetneqq \mathbb{R}_{RevCRN}^{2s}$
    \label{cor:revcrn1s_revcrn2s}
\end{corollary}

\begin{proof}
    The corollary follows directly from the proof of Proposition \ref{prop:ALG_RevCRN_nequiv}.
\end{proof}

\begin{corollary}
    $\mathbb{R}^{1s}_{RevCRN}\subsetneqq \mathbb{R}_{RTCRN}$ and thus $\mathbb{R}_{RevCRN}\cap \mathbb{R}_{RTCRN}\neq\phi$
    \label{cor:revcrn_rtcrn}
\end{corollary}
    
\begin{proof}
    Corollary \ref{cor:revcrn1s_revcrn2s} implies that $\mathbb{R}^{1s}_{RevCRN}\subsetneqq \mathbb{R}_{RevCRN}$. Further from \cite{Huang2019}, \cite{Fletcher2025_ALG_LCRN} and corollary \ref{cor:alg_equivalence} we can summarize the known inclusions as follows: $\mathbb{Q}\subsetneqq \mathbb{R}_{RevCRN}^{1s}=ALG=\mathbb{R}_{LCRN} \subsetneqq \mathbb{R}_{RTCRN}$. The rest of the corollary statement directly follows from these results.
\end{proof}

These results indicate a clear separation between the family of real numbers computed by 1 and 2-species RevCRNs, leading to the following general hierarchical separation problem.

    \begin{problem*}[$\mathbb{R}_{\mathrm{RevCRN}}$ Hierarchy] 
Recalling Definition~\ref{def:n_species_revcrn} and Proposition~\ref{prop:nested_inclusion_chain}, we have established that the sets of species-restricted computable reals satisfy the nested inclusion chain $\mathbb{R}_{\mathrm{RevCRN}}^{n\mathrm{s}} \subseteq \mathbb{R}_{\mathrm{RevCRN}}^{(n+1)\mathrm{s}}$ for all $n \in \mathbb{N}_{>0}$. We pose the following open questions regarding the strictness of these inclusions:
\begin{enumerate}
    \item Does strict inclusion hold globally, meaning $\mathbb{R}_{\mathrm{RevCRN}}^{n\mathrm{s}} \subset \mathbb{R}_{\mathrm{RevCRN}}^{(n+1)\mathrm{s}}$ for all $n \in \mathbb{N}_{>0}$?
    \item Alternatively, does the hierarchy collapse at some finite dimension, meaning there exists an index $k \in \mathbb{N}_{>0}$ such that $\mathbb{R}_{\mathrm{RevCRN}}^{n\mathrm{s}} = \mathbb{R}_{\mathrm{RevCRN}}^{(n+1)\mathrm{s}}$ for all $n \geq k$?
\end{enumerate}
If strict inclusion holds universally across all $n$, it reveals a well-defined infinite hierarchy of real numbers partitioned strictly by the minimum number of species required for their computation using RevCRNs.
\label{prob:revcrn_hierarchy}
\end{problem*}

The answer to this problem is hard. But we can initiate our investigation by examining whether such a hierarchy exists for the parameterized constraints of these ODE system classes. The following section provides a detailed discussion on this subject.

\subsection{Construction of 3-dimensional Constrained Hierarchy}

To begin with, we consider $\mathbb{S}\subseteq\mathbb{N}$ and $\mathbb{K}\subseteq\mathbb{Q}$ the set of stoichiometric coefficients and rate constants of any RevCRN system under study. We define a constraint parameter $\alpha\in\mathbb{N}$, such that the following holds
\begin{equation}
    \sum_{s\in\mathbb{S}}s+\sum_{k\in\mathbb{K}}k\leq\alpha
\end{equation}

The parameter $\alpha$ imposes a structural constraint that induces a specific subclass of RevCRNs. Such constraints can be applied to systems regardless of their species cardinality. In particular, we let $\alpha_n$ denote the constraint parameter corresponding to an $n$-species RevCRN system, where $n \in \mathbb{N}_{>0}$. We consider a three-dimensional parameter space defined over $\mathbb{N}_{>0}^2 \times \mathbb{Q}_{\geq 0}$, where each dimension corresponds to a key system property: the maximum degree of the rate equation polynomials $d \in \mathbb{N}_{>0}$, the total number of species $|S| \in \mathbb{N}_{>0}$ participating in the reaction network, and the constraint parameter $\alpha \in \mathbb{Q}_{\geq 0}$. Note that, like $\alpha_n$, we will also denote $d_n$ with $n\in\mathbb{N}$ to be the degree parameter corresponding to the $n-$species RevCRN system. Within this coordinate framework, we formally pose the constrained hierarchy problem as follows:

\begin{problem*}[Constrained $\mathbb{R}_{RevCRN}$ Hierarchy]
Characterize the internal hierarchy of RevCRN-computable real numbers over the parameter space $\mathbb{N}_{>0}^2 \times \mathbb{Q}_{\geq 0}$ under localized parameter constraints. Specifically:
\begin{enumerate}
    \item Determine whether fixing at least one parameter—thereby confining the systems to a two-dimensional hyperplane—induces a strict nested hierarchy of computable sets as the remaining parameters are varied across that plane.
    \item Establish whether the size of the resulting sets of real numbers monotonically increases as a function of the varying parameters along any given hyperplane.
\end{enumerate}
\label{prob:constrained_hierarchy}
\end{problem*}


Let $\alpha_3 \in \mathbb{Q}$ be an arbitrary constraint parameter. For any choice of parameters $\alpha_2 \in \mathbb{Q}$ and $d_2 \in \mathbb{N}$ satisfying $\alpha_2 < \alpha_3$, we let $\mathbb{R}_{2s}^{\alpha_2,d_2}$ denote the set of real numbers computable by a $2$-species RevCRN. To evaluate the set expansion under an invariant degree constraint, we explicitly fix $d_3 = d_2$ and denote the corresponding $3$-species computable set as $\mathbb{R}_{3s}^{\alpha_3,d_2}$. We initiate our discussion on hierarchical separation with the following result.

\begin{lemma}\label{lemma:base_case}
There exists a real number $r \in \mathbb{R}_{RevCRN}$ such that:
\begin{equation*}
r \in \mathbb{R}_{3s}^{\alpha_3,d_2} \setminus \mathbb{R}_{2s}^{\alpha_2,d_2}
\end{equation*}
\end{lemma}

\begin{proof}
Let $X$ and $Y$ be species participating in an arbitrary two-species reversible chemical reaction network governing the concentration trajectories $x(t)$ and $y(t)$ according to the general system in Eq.~\ref{eq:2sp_general}. We assume these networks converge asymptotically to stable fixed points, yielding the equilibrium concentrations $x_\infty = \lim_{t\rightarrow\infty} x(t)$ and $y_\infty = \lim_{t\rightarrow\infty} y(t)$. We examine the relationship between the computable classes $\mathbb{R}_{2s}^{\alpha_2,d_2}$ and $\mathbb{R}_{3s}^{\alpha_3,d_2}$ by fixing the parameters $\alpha_2 \in \mathbb{Q}$ and $d_2 \in \mathbb{N}$ such that $\alpha_2, d_2 < \infty$. Note that for any $n\in\mathbb{N}$, a choice of $\alpha_n$ inherently induces an upper bound on the maximum polynomial degree $d_n$. 

To demonstrate strict separation, we augment this system with a third species, $Z$, via a network expansion where $X$ and $Y$ act as pure catalysts:
\begin{equation} 
Z+Y \xrightleftharpoons[k_2]{k_1} Y+Z+Z, \qquad Z \xrightleftharpoons[k_4]{k_3} \emptyset, \qquad Z+X \xrightleftharpoons[k_6]{k_5} X+Z+Z 
\label{reac:ladder} 
\end{equation}
Because $X$ and $Y$ function exclusively as catalysts in reaction network~\ref{reac:ladder}, their underlying rate equations and trajectories remain unaltered. Consequently, the augmented network preserves the exact dynamics of any two-species system computing a value in $\mathbb{R}_{2s}^{\alpha_2,d_2}$, while contributing an additional ODE to the system:
\begin{equation} 
\frac{dz}{dt} = (k_1y + k_5x)z - (k_2y + k_6x)z^2 + k_4 - k_3z 
\label{eq:ode_3s} 
\end{equation}
While tracking the solution of Eq.~\ref{eq:ode_3s} is non-trivial, its asymptotic behavior can be evaluated at the steady-state limit. Setting $\frac{dz}{dt} = 0$ as $t \rightarrow \infty$ yields:
\begin{equation} 
\begin{split} 
\lim_{t\rightarrow\infty} \left[ (k_1y + k_5x)z - (k_2y + k_6x)z^2 + k_4 - k_3z \right] &= \lim_{t\rightarrow\infty}\frac{dz}{dt} = 0 \\ 
(k_1y_\infty + k_5x_\infty) z_\infty - (k_2y_\infty + k_6x_\infty) z_\infty^2 + k_4 - k_3z_\infty &= 0 
\end{split} 
\label{eq:3rd_specie_rate_limit} 
\end{equation}
For any initialization resulting in $x_\infty > 0$ or $y_\infty > 0$, solving this quadratic equation yields the steady-state roots:
\begin{equation} 
z_\infty = \frac{(k_1y_\infty + k_5x_\infty - k_3) \pm \sqrt{(k_1y_\infty + k_5x_\infty - k_3)^2 + 4(k_2y_\infty + k_6x_\infty)k_4}}{2(k_2y_\infty + k_6x_\infty)} 
\label{eq:z_solution} 
\end{equation}
Physically viable concentrations require $z_\infty$ to be real and non-negative ($z_\infty \in \mathbb{R}_{\geq 0}$). The existence of a stable attractor within this domain is structurally guaranteed by the signs of the vector fields in Eq.~\ref{eq:ode_3s}. Specifically, evaluated at the boundary $z = 0$, we find $\frac{dz}{dt} = k_4 > 0$ whenever $k_4 > 0$. Conversely, for sufficiently large values of $z$, the quadratic term dominates, ensuring that $\frac{dz}{dt} < 0$. Since the rate equation is continuous, by the Intermediate Value Theorem \cite{realanalysisintroduction}, this sign reversal confirms the presence of at least one stable fixed point.

Let $X_\infty$ and $Y_\infty$ denote the complete sets of fixed points attainable by species $X$ and $Y$, respectively, across the entire family of valid two-species networks, such that $\mathbb{R}_{2s}^{\alpha_2,d_2} = X_\infty \cup Y_\infty$. Because $\alpha_2 < \infty$, the maximum stoichiometric values and rate constants are strictly bounded ($\max \mathbb{S}, \max \mathbb{K} < \infty$). We establish the strict expansion of the computable space via the following claims.

\begin{claim} 
For any $x_\infty \in X_\infty$ and $y_\infty \in Y_\infty$, there always exists a valid selection of rate constants such that the augmented network~\ref{reac:ladder} produces an equilibrium concentration satisfying $z_\infty > x_\infty$ and $z_\infty > y_\infty$, where $z_\infty \in Z_\infty \subseteq \mathbb{R}_{3s}^{\alpha_3,d_2}$.
\label{claim:strict_increasing} 
\end{claim} 

\begin{claimproof} 
We evaluate the positive root of Eq.~\ref{eq:z_solution} under the target inequalities by examining two distinct directional bounds:
\begin{enumerate} 
    \item \textbf{Case 1 ($z_\infty > x_\infty$):} Squaring and rearranging the root expression reveals that this inequality is structurally satisfied whenever the network parameters satisfy:
    \begin{equation} 
    x_\infty^2(k_2y_\infty + k_6x_\infty) - x_\infty(k_1y_\infty + k_5x_\infty - k_3) - k_4 < 0 
    \label{eq:x_infty} 
    \end{equation} 
    \item \textbf{Case 2 ($z_\infty > y_\infty$):} Following an identical algebraic expansion, the condition $z_\infty > y_\infty$ holds true whenever:
    \begin{equation} 
    y_\infty^2(k_2y_\infty + k_6x_\infty) - y_\infty(k_1y_\infty + k_5x_\infty - k_3) - k_4 < 0 
    \label{eq:y_infty} 
    \end{equation} 
\end{enumerate} 
By Lemma~\ref{lem:bound}, the finiteness of the structural parameters ensures that the sets $X_\infty$ and $Y_\infty$ are bounded, implying that their suprema and infima exist. Consequently, for any arbitrary fixed pair of equilibrium values $x_\infty$ and $y_\infty$, there always exists a rational parameter configuration $k_1, \dots, k_6 \in \mathbb{Q}_{\geq 0}$ that satisfies both quadratic inequalities simultaneously. This confirms the existence of a reachable steady state $z_\infty \in Z_\infty$ that strictly bounds the baseline coordinate pair.
\end{claimproof}

By assumption, let $\sup X_\infty \leq \sup Y_\infty$ without loss of generality. For any arbitrary small bounds $\varepsilon_1, \varepsilon_2 >0$, we can select baseline trajectories approaching their limits such that $y_\infty \in [\sup Y_\infty - \varepsilon_1, \sup Y_\infty)$ and $x_\infty \in [\sup X_\infty - \varepsilon_2, \sup X_\infty)$. By applying Claim~\ref{claim:strict_increasing}, the corresponding augmented network produces an equilibrium state satisfying $z_\infty > \sup Y_\infty \geq y_\infty$. Because $z_\infty$ strictly exceeds the supremum of the baseline space, it follows that $z_\infty \notin \mathbb{R}_{2s}^{\alpha_2,d_2}$. 

To complete the separation proof under the structural budget framework, we establish how this network modification impacts the parameter space constraint.

\begin{claim} 
The structural constraint parameters satisfy $\alpha_3 > \alpha_2$, where $\alpha_2, \alpha_3 \in \mathbb{Q}_{\geq 0}$. 
\end{claim} 

\begin{claimproof} 
Let $\mathbb{S}_2, \mathbb{S}_3$ and $\mathbb{K}_2, \mathbb{K}_3$ denote the sets of stoichiometric coefficients and reaction rate constants for the two-species and three-species systems, respectively, bounded such that:
\begin{equation*} 
\sum_{s\in\mathbb{S}_2}s + \sum_{k\in\mathbb{K}_2}k \leq \alpha_2 \qquad \text{and} \qquad \sum_{s\in\mathbb{S}_3}s + \sum_{k\in\mathbb{K}_3}k \leq \alpha_3 
\end{equation*}
The introduction of the third species $Z$ requires appending the additional reactions specified in network~\ref{reac:ladder} to the system. This expansion introduces a non-empty set of strictly positive stoichiometric coefficients $\mathbb{S}^+$ and positive rational rate constants $\mathbb{K}^+$, such that the total structural compositions expand to $\mathbb{S}_3 = \mathbb{S}_2 \cup \mathbb{S}^+$ and $\mathbb{K}_3 = \mathbb{K}_2 \cup \mathbb{K}^+$. Summing over the newly combined indices yields:
\begin{equation*} 
\begin{split} 
\sum_{s\in\mathbb{S}_3}s + \sum_{k\in\mathbb{K}_3}k &= \sum_{s_2\in\mathbb{S}_2}s_2 + \sum_{k_2\in\mathbb{K}_2}k_2 + \sum_{s^+\in\mathbb{S}^+}s^+ + \sum_{k^+\in\mathbb{K}^+}k^+ \\ 
&> \sum_{s_2\in\mathbb{S}_2}s_2 + \sum_{k_2\in\mathbb{K}_2}k_2 
\end{split} 
\end{equation*}
It follows directly from the lower-bound maximization constraint that:
\begin{equation*}
\alpha_2 < \sum_{s_2\in\mathbb{S}_2}s_2 + \sum_{k_2\in\mathbb{K}_2}k_2 + \sum_{s^+\in\mathbb{S}^+}s^+ + \sum_{k^+\in\mathbb{K}^+}k^+ \leq \alpha_3
\end{equation*}
which proves $\alpha_2 < \alpha_3$, thereby concluding the proof of the lemma.
\end{claimproof}
\end{proof}

\vspace{-0.35cm}

The hierarchical separation established in Lemma~\ref{lemma:base_case} generalizes inductively across higher species counts. Let $\alpha_{n+1} \in \mathbb{Q}$ be an arbitrary constraint parameter for an $(n+1)$-species network. For any choice of parameters $\alpha_n \in \mathbb{Q}$ and $d_n \in \mathbb{N}$ satisfying $\alpha_n < \alpha_{n+1}$, we define $\mathbb{R}_{ns}^{\alpha_n,d_n}$ as the $n$-species computable set. Enforcing an invariant degree constraint across dimensions by setting $d_{n+1} = d_n$, we define the corresponding $(n+1)$-species computable set as $\mathbb{R}_{(n+1)s}^{\alpha_{n+1},d_n}$. This yields the following general theorem.

\begin{theorem}\label{thm:general_case_expansion}
There exists a real number $r \in \mathbb{R}_{RevCRN}$ such that:
\begin{equation}
r \in \mathbb{R}_{(n+1)s}^{\alpha_{n+1},d_n} \setminus \mathbb{R}_{ns}^{\alpha_n,d_n}
\end{equation}
\end{theorem}


\begin{proof}
    In this case, we will use a general construction similar to the RevCRN construction described earlier. Here, we will use a set of species $S=\{X_j|j\in\mathbb{N},1\leq j\leq n+1\}$. Consider the follows RevCRN construction.
    \begin{equation}
        \begin{split}
            R1&:X_{n+1}+X_{1}\xrightleftharpoons[k_{12}]{k_{11}}X_{n+1}+X_{n+1}+X_{1}\\
            \vdots\\
            Rn&:X_{n+1}+X_{n}\xrightleftharpoons[k_{(n)2}]{k_{(n)1}}X_{n+1}+X_{n+1}+X_{n}\\
            R(n+1)&:X_{n+1}\xrightleftharpoons[k_{(n+1)2}]{k_{(n+1)1}}\phi\\
        \end{split}
    \end{equation}

    The mass-action kinetics of the above RevCRN will induce the following general ODE
    \begin{equation*}
        \frac{dX_{n+1}}{dt}=k_{(n+1)2}+\bigg(\sum_{i=1}^{n}k_{i1}x_i-k_{(n+1)1}\bigg)x_{n+1}-\bigg(\sum_{i=1}^{n}k_{i2}x_i\bigg)x_{n+1}^2
    \end{equation*}

In this case, we again get back a general quadratic equation with $a=-\lim_{t\rightarrow\infty}\bigg(\sum_{i=1}^{n}k_{i2}x_i\bigg)$, $b=\lim_{t\rightarrow\infty}\sum_{i=1}^{n}k_{i1}x_i-k_{(n+1)1}$, and $c=k_{(n+1)2}$ at the limit $t\rightarrow\infty$. On solving the quadratic equation in a similar way as done before, we get back the solution for $x_{(n+1)\infty}$ to be $\frac{-b\pm\sqrt{b^2-4ac}}{2a}$.

Note that claim \ref{claim:strict_increasing} can be extended to the general case as well. We can get $x_{n+1}>x_i,\forall i\in \mathbb{N}, 1\leq i\leq n$ whenever
\begin{equation}
    ax_i^2-bx_i-c<0
\end{equation}
It can be easily observed that, similar to the claim \ref{claim:strict_increasing}, there exists at least one combination of positive rational rate constants such that the above inequality holds. This implies that $\exists x_{(n+1)\infty}\in\mathbb{R}$ such that $x_{(n+1)\infty}\notin\mathbb{R}_{ns}^{\alpha_n,d_n}$.

Finally, we can construct the sets $\mathbb{S}_{n+1}=\mathbb{S}_n\cup\mathbb{S}^+$ and $\mathbb{K}_{n+1}=\mathbb{K}_n\cup\mathbb{K}^+$ and show that $\alpha_n<\alpha_{n+1}$.
\end{proof}


Let $\alpha_n \in \mathbb{Q}$ and $d_n \in \mathbb{N}$ be arbitrary constraints governing an $n$-species network, defining the computable set $\mathbb{R}_{ns}^{\alpha_n,d_n}$. We evaluate the increase in the computational capability of RevCRNs strictly due to an increase in the value of the constraint parameter by introducing a higher constraint bound $\alpha'_n \in \mathbb{Q}$ satisfying $\alpha'_n > \alpha_n$. Enforcing a strictly invariant parameter configuration for both the species count $n$ and the polynomial degree $d_n$, we denote the expanded computable set as $\mathbb{R}_{ns}^{\alpha'_n,d_n}$. This yields the following strict structural separation result.

\begin{proposition}\label{prop:constraint_expansion_separation}
There exists a real number $r \in \mathbb{R}_{RevCRN}$ such that:
\begin{equation}
r \in \mathbb{R}_{ns}^{\alpha'_n,d_n} \setminus \mathbb{R}_{ns}^{\alpha_n,d_n}
\end{equation}
\end{proposition}

\begin{proof}
    According to Lemma \ref{lem:bound}, we understand that with the general form of a reversible reaction network (equation \ref{eq:rev_rate_general_nspecies}), computing real numbers, all species concentrations are bounded as long as there is no uncontrolled mass accumulation. In particular, $\exists C_1\in \mathbb{R}$ such that $\forall t\in[0,\infty),||x(t)||_\infty<C_1$. From the proof of Lemma \ref{lem:bound}, we know that the value $C_1$ depends on the values present in the set of stoichiometric coefficients $\mathbb{S}$ and the rate constants $\mathbb{K}$.
    
    Let $\mathcal{C}_{\alpha_n}$ be the set of all possible bounds $C_1$ induced by the constraint $\alpha_n$. It can be easily verified from \eqref{eq:c1b} that for $\alpha'_n>\alpha_n$, $\max \mathcal{C}_{\alpha'_n}>\max\mathcal{C}_{\alpha_n}$. From Lemma \ref{lem:rational}, we know that the set of rational numbers, $\mathbb{Q}$, can be computed by a single-species RevCRN. Furthermore, from Proposition \ref{prop:nested_inclusion_chain}, we can conclude that $\mathbb{Q}$ can be computed by an $n$-species RevCRN for any $n \in \mathbb{N}$ with $n > 0$. 

    Given that $\mathbb{Q}$ is dense, there exists a rational number $q \in \mathbb{Q}$ such that $\max \mathcal{C}_{\alpha'_n} > q > \max \mathcal{C}_{\alpha_n}$. This implies that $q \notin \mathbb{R}_{ns}^{\alpha_n,d_n}$, but $q \in \mathbb{R}_{ns}^{\alpha'_n,d_n}$.
\end{proof}


We evaluate the increase in computational capability of RevCRNs strictly due to an increase in the value of the degree parameter while keeping the constraint parameter $\alpha_n \in \mathbb{Q}$ and the number of species $n\in\mathbb{N}_{>0}$ fixed. To do this, we introduce a higher degree bound $d'_n \in \mathbb{N}$ with respect to the baseline degree $d_n \in \mathbb{N}$ satisfying $d'_n > d_n$. Under these parameters, we define $\mathbb{R}_{ns}^{\alpha_n,d_n}$ as the set of all real numbers computable by an $n$-species RevCRN subject to the constraints $(\alpha_n, d_n)$. Similarly, let $\mathbb{R}_{ns}^{\alpha_n,d'_n}$ denote the set of all real numbers computable by an $n$-species RevCRN under the constraints $(\alpha_n, d'_n)$. Then we can define the following proposition.

\begin{proposition}\label{prop:degree_expansion_separation}
There exists an $r \in \mathbb{R}_{RevCRN}$ such that
\begin{equation}
    r \in \mathbb{R}_{ns}^{\alpha_n,d'_n} \setminus \mathbb{R}_{ns}^{\alpha_n,d_n}
\end{equation}
\end{proposition}

\begin{proof}
    According to the standard mathematical properties of single and multivariate polynomials, it is established that for a fixed number of variables and bounded rational coefficients, a polynomial of degree \(d + 1\) can compute a larger set of algebraic numbers than one of degree \(d\).
\end{proof}

\section{Discussion}

The behavior of Shannon's General Purpose Analog Computers (GPACs) is governed by ODEs, similar to the dynamics of CRNs. Notably, it has been established that CRNs represent a specific subclass of GPACs, as documented in the literature \cite{Bournez2021_survey_analog, huang_phd_thesis, GPAC_RTCRN_equivalence}. Prior investigations have demonstrated that polynomial ODEs, alongside GPACs, possess the capability to simulate any arbitrary Turing machine \cite{BOURNEZ2016106_compute_pode, GRACA_PIVP}. In a recent study in 2019, Huang, Klinge, and Lathrop 
show that the class of real numbers computable in real time by both CRNs and GPACs is equivalent  \cite{GPAC_RTCRN_equivalence}. In this context, the proposed work highlights the existence of reversible computing in the broad field of analog computation.

This work formally defines a computable version of the Reversible Chemical Reaction Network (RevCRN) and explores its ability to compute real numbers. The results presented in this paper expand the understanding of RevCRN-computable real numbers $\mathbb{R}_{RevCRN}$. It further explores the overlap between $\mathbb{R}_{RevCRN}$ and other CRN-computable reals, such as $\mathbb{Q}, ALG, \mathbb{R}_{LCRN}$ and $\mathbb{R}_{RTCRN}$, previously introduced in the literature.

The results of this work show that 1-species RevCRNs can compute any rational and algebraic number. In particular, $\mathbb{Q}\subsetneqq ALG=\mathbb{R}_{RevCRN}^{1s}$. This inclusion further confirms that the 1-species RevCRN can compute any numbers in $\mathbb{R}_{LCRN}$ as $ALG=\mathbb{R}_{LCRN}$ while guaranteeing the overlap with $\mathbb{R}_{RTCRN}$, i.e., $\mathbb{R}_{RevCRN}\cap\mathbb{R}_{RTCRN}\neq \phi$. Upon further investigation, our results confirm that $ALG$ and $\mathbb{R}_{LCRN}$ are strict subsets of $\mathbb{R}_{RevCRN}$. Our finding supports the fact that $\mathbb{R}_{RevCRN}$ can also compute transcendental numbers. Furthermore, it was also found that the set of real numbers computable by detailed-balanced RevCRNs is a subset of $ALG$, i.e., $\mathbb{R}_{RevCRN}^{DetBal} \subseteq ALG$. However, in this work, we did not investigate the exact inclusion between $\mathbb{R}_{RevCRN}$ and $\mathbb{R}_{RTCRN}$, leaving it as an open problem for future work. This can be stated precisely as follows.

\begin{openproblem*}
We know that $\mathbb{R}_{RevCRN}\cap\mathbb{R}_{RTCRN}\neq \phi$. But it is open to investigate whether 
\begin{enumerate}
    \item $\mathbb{R}_{RevCRN}\subsetneqq\mathbb{R}_{RTCRN}$
    \item $\mathbb{R}_{RTCRN}\subsetneqq\mathbb{R}_{RevCRN}$
    \item $\mathbb{R}_{RevCRN}=\mathbb{R}_{RTCRN}$
    \item None of the above is true but $\mathbb{R}_{RevCRN}$ and $\mathbb{R}_{RTCRN}$ have a significant overlap
\end{enumerate}
\end{openproblem*}

In section \ref{sec:hierarchy}, we demonstrated the existence of a 3-dimensional hierarchy of computable reals within $\mathbb{R}_{RevCRN}$. The hierarchy indicates that the size of a set of real numbers increases with an increase in one of three variables: the parameter bound \(\alpha\), the number of species \(|S|\), or the degree of the polynomial \(d\), while fixing the other two variables as a constant. The results provided are a consequence of investigating a more general question of whether $\mathbb{R}_{RevCRN}^{ns}\subset\mathbb{R}_{RevCRN}^{(n+1)s}$ or whether $\exists k\in\mathbb{N}$ such that the hierarchy collapses at the $k^{th}$ level. In light of the preceding analysis, we conclude by proposing the following conjecture and presenting it as another open problem for future exploration.

\begin{conjecture*}
For any $n\in\mathbb{N}$, we have $\mathbb{R}_{RevCRN}^{ns}\subset\mathbb{R}_{RevCRN}^{(n+1)s}$ where, $\mathbb{R}_{RevCRN}^{ns}$ and $\mathbb{R}_{RevCRN}^{(n+1)s}$ are the set of RevCRN-computable reals using $n$ and $n+1$ species respectively.
\end{conjecture*}

In conclusion, future research can explore a relaxed notion of reversibility, known as weak reversibility, and examine its computational capabilities. In weakly reversible systems, although there is no guarantee of energy optimization, we can still achieve CRN reusability and revert the system to any state, including the initial state, via specific reaction pathways. This approach also simplifies species extraction, similar to fully reversible networks. By relaxing the strict reversibility condition, we expect to expand the set of computable numbers for weakly reversible chemical reaction networks ($\mathbb{R}_{WRevCRN}$). 

We know from the standard definitions of fully reversible and weakly reversible reaction networks that all fully reversible reaction networks are also weakly reversible. However, the reverse is not true \cite{feinberg2019foundations}. Let $\mathbb{R}_{RevCRN}^{FR}$ represent the set of real numbers that can be computed by fully reversible reaction networks (RevCRNs), where every forward reaction is paired with a backward reaction bearing a positive rate constant. By definition, this implies that $\mathbb{R}_{RevCRN}^{FR} \subseteq \mathbb{R}_{WRevCRN}$. However, further results on inclusion and the exact structure of $\mathbb{R}_{WRevCRN}$ remain an open problem.

\bibliography{reference}

\appendix

\section*{Appendix}

\setcounter{equation}{0}
\renewcommand{\theequation}{A.\arabic{equation}}
\setcounter{note}{0}
\setcounter{remark}{0}

\section{General form of rate equations of Reversible CRNs}
\label{sec:general_eq_123_sp}
\subsection{1 Specie Reaction}

Consider the following general equation for a 1-species reaction network.

\begin{equation}
    R_i:\qquad a_iX\xrightleftharpoons[k_{i2}]{k_{i1}}b_iX,\qquad i\in \mathbb{N}, 1\leq i\leq |R|
    \label{eq:1s_rev_reac}
\end{equation}

Here $R$ is the set of reactions in the reaction network and $|R|$ represents the cardinality of $R$. $a_i\in\mathbb{N}$ and $b_i\in\mathbb{N}$ are constants representing the stoichiometric coefficients, where $k_{i1}\in\mathbb{Q}$ and $k_{i2}\in\mathbb{Q}$ are the rate constants of the forward and backward reactions, respectively. These notations remain consistent for the rest of the document.

\begin{note}
    There can be reactions that involve 0-complex species, specifically when $\phi$ is one of the complexes. In these types of reactions, the reaction network is influenced by external factors. In reality, reversibility applies only to reactions that include both reactant and product species. However, the interaction of a single species with an external factor is not typically considered in the study of reversible reactions.
\end{note}

For example, an external factor, such as human intervention, can introduce a new species or additional mass of an existing species into the reacting system. In this case, it is not necessary to remove mass from the system; instead, one can allow the system to remain in its current state after the mass injection and let the reaction progress through the network. Therefore, external control does not need to maintain reversibility when adding or removing mass from the system. 

It is essential, however, that the rate constants for addition and removal be positive rational numbers to avoid the potential for computational advantages. Conversely, reactions that involve internal interactions with other species, which possess non-zero complexes on both the reactant and product sides, must adhere strictly to the condition that if there is a forward reaction, there must be a corresponding backward reaction. This means that if \( k_{i1} > 0 \), then it must also be true that \( k_{i2} > 0 \) and vice versa. These conditions will remain consistent throughout the rest of the document.

\subsection*{ODE Analysis}

The general rate equation of the chemical reactions provided in \ref{eq:1s_rev_reac} is given as follows

\begin{equation}
        \frac{dx}{dt}=\sum_{i=1}^{|R|}(b_i-a_i)\big(k_{i1}x^{a_i}-k_{i2}x^{c_i}\big)
        \label{eq:1sp_rate_equation}
\end{equation}

\subsection{2 Species Reaction}

Consider the following general equation for a 2-species reaction network.

\begin{equation}
    R_i:\qquad a_iX+b_iY\xrightleftharpoons[k_{i2}]{k_{i1}}c_iX+d_iY,\qquad i\in \mathbb{N}, 1\leq i\leq |R|
    \label{eq:2s_rev_reac}
\end{equation}

Here, $R$ is the same as in the last case. $a_i,b_i,c_i$ and $d_i$ are constants representing the stoichiometric coefficients where $k_{i1}$ and $k_{i2}$ are the rate constants of the forward and backward reactions, respectively. Similar to the case of 1-species reaction networks, we have a set of general reactions indexed by $i$.

\subsection*{ODE Analysis}

The general rate equation of the chemical reactions provided in \ref{eq:2s_rev_reac} is given as follows

\begin{equation}
    \begin{split}
        \frac{dx}{dt}&=\sum_{i=1}^{|R|}(c_i-a_i)\big(k_{i1}x^{a_i}y^{b_i}-k_{i2}x^{c_i}y^{d_i}\big)\\
        \frac{dy}{dt}&=\sum_{i=1}^{|R|}(d_i-b_i)\big(k_{i1}x^{a_i}y^{b_i}-k_{i2}x^{c_i}y^{d_i}\big)
    \end{split}
    \label{eq:2sp_general}
\end{equation}

Note that if $b_i=d_i, \frac{dy}{dt}=0\Rightarrow y(t)=constant$. As per construction $c_i>a_i$ and $d_i>b_i$.

Further, if $constant=1$, we recover the general equation for a 1-species reaction network. Thus, $\mathbb{R}_{RevCRN}^{1sp}\subseteq\mathbb{R}_{RevCRN}^{2s}$

\subsection{3 Species Reaction}
\label{sec:3s_revcrn}

Consider the following general equation for a 3-species reaction network.

\begin{equation}
    R_i:\qquad a_iX+b_iY+c_iZ\xrightleftharpoons[k_{i2}]{k_{i1}}d_iX+e_iY+f_iZ,\qquad i\in \mathbb{N}, 1\leq i\leq |R|
    \label{eq:3s_rev_reac}
\end{equation}

The fundamental constraint and notations follow from the last section.

\subsection*{ODE Analysis}

The general rate equation of the chemical reactions provided in \ref{eq:2s_rev_reac} is given as follows

\begin{equation*}
    \begin{split}
        \frac{dx}{dt}&=\sum_{i=1}^{|R|}(d_i-a_i)\big(k_{i1}x^{a_i}y^{b_i}z^{c_i}-k_{i2}x^{d_i}y^{e_i}z^{f_i}\big)\\
        \frac{dy}{dt}&=\sum_{i=1}^{|R|}(e_i-b_i)\big(k_{i1}x^{a_i}y^{b_i}z^{c_i}-k_{i2}x^{d_i}y^{e_i}z^{f_i}\big)\\
        \frac{dz}{dt}&=\sum_{i=1}^{|R|}(f_i-c_i)\big(k_{i1}x^{a_i}y^{b_i}z^{c_i}-k_{i2}x^{d_i}y^{e_i}z^{f_i}\big)
    \end{split}
\end{equation*}

Like in the last case, it can be easily shown that $\mathbb{R}_{RevCRN}^{2s}\subseteq\mathbb{R}_{RevCRN}^{3s}$. The details are thus skipped in this analysis.

\section{Proof of Lemma \ref{lem:bound}}
\label{proof:bound}

We will prove this result through a series of estimates and by contradiction. We demonstrate this using equation \eqref{eq:2sp_general}. The general case can be derived from this.

We assume WLOG that the polynomial in the RHS of \eqref{eq:2sp_general} has powers that are in ascending order. Let, $\forall i\in\mathbb{N}$,
\begin{equation*}
    \begin{split}
        m_1&=\min\limits_{1\leq i\leq |R|}{[(c_{i}-a_{i})k_{i2}]},\qquad M_1=\max\limits_{1\leq i\leq |R|}{[(c_{i}-a_{i})k_{i1}]},\\ m_2&=\min\limits_{1\leq i\leq |R|}{[(d_{i}-b_{i})k_{i2}]},\qquad  M_2=\max\limits_{1\leq i\leq |R|}{[(d_{i}-b_{i})k_{i1}]}
    \end{split}
\end{equation*}

We next split the summation in \eqref{eq:2sp_general} into two separate summations, and use properties of maxima and minima to yield,
\begin{equation}
    \begin{split}
        \frac{dx}{dt} + |R|m_1x^{c_{|R|}}y^{d_{|R|}}&\leq|R|M_1x^{a_{|R|}}y^{b_{|R|}}\\
        \frac{dy}{dt}+|R|m_2x^{c_{|R|}}y^{d_{|R|}}&\leq|R|M_2x^{a_{|R|}}y^{b_{|R|}}\\
    \end{split}
    \label{eq:2sp_general1}
\end{equation}

Now a time scaling of the form $\tau=|R|t$, yields,

\begin{equation}
    \begin{split}
        \frac{dx}{d\tau} + m_1y^{d_{|R|}}&\leq M_1x^{a_{|R|}}y^{b_{|R|}}\\
        \frac{dy}{d\tau}+ m_2y^{d_{|R|}}&\leq M_2x^{a_{|R|}}y^{b_{|R|}}\\
    \end{split}
    \label{eq:2sp_general12}
\end{equation}

Note, $x^{c_{|R|}}y^{d_{|R|}} = \left(\epsilon x^{a_{|R|}}y^{b_{|R|}}\right)^{\frac{c_{|R|}d_{|R|}}{a_{|R|}b_{|R|}}} \left(\frac{1}{\epsilon}\right)^{\frac{c_{|R|}d_{|R|}}{a_{|R|}b_{|R|}}}$

Next we use Young's inequality \cite{sell2002dynamics} with $\epsilon$ ($ab \leq \frac{a^{p}}{p}+\frac{b^{q}}{q}, \frac{1}{p} + \frac{1}{q} = 1$), with $p=\frac{c_{|R|}d_{|R|}}{a_{|R|}b_{|R|}}, q=\frac{\frac{c_{|R|}d_{|R|}}{a_{|R|}b_{|R|}}}{\frac{c_{|R|}d_{|R|}}{a_{|R|}b_{|R|}} - 1}$, on the RHS of \eqref{eq:2sp_general12}, and choose $\epsilon$ such that, 

\begin{equation}
\label{eq:ep1}
\frac{1}{2}\min{[(c_{i}-a_{i})k_{i2}]} = \max{[(c_{i}-a_{i})k_{i1}]} \frac{\epsilon^{\frac{c_{|R|}d_{|R|}}{a_{|R|}b_{|R|}}}}{\frac{c_{|R|}d_{|R|}}{a_{|R|}b_{|R|}}}
\end{equation}

to yield,

\begin{equation}
    \begin{split}
        \frac{dx}{d\tau} + \frac{1}{2}\min{[(c_{i}-a_{i})k_{i2}]}x^{c_{|R|}}y^{d_{|R|}}& \leq C_{1}\\
        \frac{dy}{d\tau}+ \frac{1}{2}\min{[(c_{i}-a_{i})k_{i2}]}x^{c_{|R|}}y^{d_{|R|}}&\leq C_{1}\\
    \end{split}
    \label{eq:2sp_general123}
\end{equation}
Here,
\begin{equation}
\label{eq:c1b}
    C_{1}= \left(\left(\frac{1}{\epsilon}\right)^{\frac{c_{|R|}d_{|R|}}{a_{|R|}b_{|R|}}}\right)^{\frac{\frac{c_{|R|}d_{|R|}}{a_{|R|}b_{|R|}}}{\frac{c_{|R|}d_{|R|}}{a_{|R|}b_{|R|}} - 1}}
\end{equation}
in which $\epsilon$ is the solution of \eqref{eq:ep1}.

We next use a contradiction argument. Assume there is no upper bound on $ x$ or $ y$. Then for any $K$, arbitrarily large, there exists a time $t_{2}$, such that for times $t > t_{2}$, 

\begin{equation}
        ||x||_{\infty} > K, \ ||y||_{\infty} > K.
    \end{equation}

Then using $||y||_{\infty} > K$ in the first equation in \eqref{eq:2sp_general123}, on the time interval $[t_{2},\infty)$ yields,

\begin{equation}
     \frac{dx}{d\tau} + \frac{1}{2}\min{[(c_{i}-a_{i})k_{i2}]}(K)^{d_{|R|}} x^{c_{|R|}}\leq C_{1}
    \label{eq:2sp_general1234}
\end{equation}

But now, by applying Gr\"onwall's lemma \cite{perko2013differential}, $x$ is bounded, invalidating the unbounded assumption on it, yielding a contradiction. For the three-species case and beyond, this approach yields the boundedness result.

\section{Analysis of ladder reaction for 4 species system}
\label{sec:4s_ladder_example}

    We can observe a similar result when we add one more reaction to the system, pushing it to a 4-species system of this form.
\begin{equation}
    \begin{split}
        W&\xrightleftharpoons[k_2]{k_1}\phi\\
        W+X&\xrightleftharpoons[k_4]{k_3}X+W+W\\
        W+Y&\xrightleftharpoons[k_6]{k_5}Y+W+W\\
        W+Z&\xrightleftharpoons[k_8]{k_7}Z+W+W
    \end{split}
    \label{reac:ladder_4sp}
\end{equation}
Here, $X$, $Y$, and $Z$ are any 3-species reaction constrained by $\alpha_3$ and independent of $W$. We then get the following ODE for W.
\begin{equation}
    \frac{dw}{dt}=k_2+(k_3x+k_5y+k_7z-k_1)w-(k_4x+k_6y+k_8z)w^2
\end{equation}

Similar to the last case, at the limit $t\rightarrow\infty$, we again get stable points $x_\infty\in X_\infty, y_\infty\in Y_\infty$, and $z_\infty\in Z_\infty$ respectively. With these fixed points, we get
\begin{equation}
\begin{split}
    w_\infty=\frac{-b\pm\sqrt{b^2-4ac}}{2a}\\
    Here,\ a=k_4x_\infty+k_6y_\infty+k_8z_\infty\\
    b=k_3x_\infty+k_5y_\infty+k_7z_\infty-k_1\\
    c=k_2
\end{split}
\end{equation}

Using a similar argument as in claim \ref{reac:ladder}, we can show that there are configurations of rate constants such that $w_\infty > x_\infty, y_\infty, z_\infty$. Further, using a similar proof as in Lemma \ref{lemma:base_case}, we can show that $w_\infty\notin \mathbb{R}_{3s}^{\alpha_e,d_e=\mathbb{N}}$ for any $d_3$ and $\alpha_3$ with $\alpha_3<\alpha_4$.

\end{document}